\documentclass[letterpaper,11pt]{article}

\usepackage[margin=1in]{geometry}
\usepackage{authblk}

\usepackage{amsmath,amssymb,amsthm}
\usepackage{tabularx}
\usepackage[algoruled,lined,algonl,procnumbered,nosemicolon]{algorithm2e}
\usepackage{lineno}
\usepackage{subcaption}
\usepackage{float}
\usepackage{comment} 
\usepackage{enumitem}
\usepackage[hyphens]{url}
\usepackage{graphicx,hyperref}
\usepackage{breakcites}
\usepackage{tikz}
\usepackage{booktabs}
\usepackage[textsize=scriptsize]{todonotes}
\usepackage{lmodern}

\usepackage{orcidlink}
\usepackage{placeins}  
\usepackage{breakcites}

\colorlet{darkorange}{orange!85!black}
\colorlet{darkblue}{blue!75!black}
\colorlet{darkgreen}{green!50!black}

\newtheorem{theorem}{Theorem}
\newtheorem{corollary}{Corollary}

\newtheorem{lemma}{Lemma}
\newtheorem{observation}{Observation}
\newtheorem{remark}{Remark}
\newtheorem{conjecture}{Conjecture}
\newtheorem{definition}{Definition}
\newtheorem{example}{Example}

\newtheorem{claim}{Claim}
\newtheorem*{claim*}{Claim}

\newenvironment{claimproof}{\noindent{\em Proof of the claim.}}{\hfill $\diamond$ \par\medskip}

\newcommand{\PWS}{\textnormal{\textsc{Pinwheel Scheduling}}}

\newcommand{\PM}{\textnormal{\textsc{Position Matching}}}

\newcommand{\tP}{\textnormal{3\textsc{-Partition}}}

\newcommand{\NMTS}{\textnormal{\textsc{Numerical Matching with Target Sums}}}
\newcommand{\NMTSshort}{\textnormal{\textsc{NMTS}}}

\newcommand{\EWPM}{\textnormal{\textsc{Exact Weighted Perfect Matching}}}
\newcommand{\EWPMshort}{\textnormal{\textsc{EWPM}}}

\newcommand{\EM}{\textnormal{\textsc{Exact Matching}}}

\newcommand{\INTDM}{\textnormal{\textsc{Inequality Numerical}\ 3\textsc{-Dimensional Matching}}}
\newcommand{\INTDMshort}{\textnormal{\textsc{IN3DM}}}

\newcommand{\RNTDM}{\textnormal{\textsc{Restricted Numerical}\ 3\textsc{-Dimensional Matching}}}
\newcommand{\RNTDMshort}{\textnormal{\textsc{RN3DM}}}

\newcommand{\SRNM}{\textnormal{\textsc{Semi-restricted Numerical Matching with Target Sums}}}
\newcommand{\SRNMshort}{\textnormal{\textsc{SRNMTS}}}

\newcommand{\kV}{\textnormal{\ensuremath{k}\textsc{-Visits}}}

\newcommand{\ILP}{\textnormal{\textsc{Integer Linear Programming}}}

\newcommand{\oneV}{\textnormal{1\textsc{-Visit}}}
\newcommand{\twoV}{\textnormal{2\textsc{-Visits}}}
\newcommand{\oneortwoV}{\textnormal{(1 \textsc{or} 2)\textsc{-Visits}}}
\newcommand{\threeV}{\textnormal{3\textsc{-Visits}}}

\def\bigO{\ensuremath{\mathcal{O}}\xspace}

\usepackage[table]{xcolor} 
\usepackage{multirow}

\makeatother



\begin{document}
\title{Hardness, Tractability and Density Thresholds of\\ finite Pinwheel Scheduling Variants}




\author[1,2]{Sotiris Kanellopoulos \orcidlink{0009-0006-2999-0580} }
\author[3]{Giorgos Mitropoulos \orcidlink{0009-0000-0479-263X} }
\author[1,2]{Christos Pergaminelis \orcidlink{0009-0009-8981-3676} }
\author[1,2]{Thanos Tolias \orcidlink{0009-0003-8354-8855} }

\affil[1]{National Technical University of Athens, Greece}
\affil[2]{Archimedes, Athena Research Center, Greece}
\affil[3]{Sorbonne Université, CNRS, LIP6, F-75005 Paris, France}

\affil[ ]{\{s.kanellopoulos, chr.pergaminelis\}@athenarc.gr, georgios.mitropoulos@lip6.fr, thanostolias@mail.ntua.gr}

\date{}

\maketitle
%
\begin{abstract}

    The $k$-\textsc{Visits} problem is a recently introduced finite version of \textsc{Pinwheel Scheduling} [Kanellopoulos et al., SODA 2026~\cite{kVisits}]. Given the deadlines of $n$ tasks, the problem asks whether there exists a schedule of length $kn$ executing each task exactly $k$ times, with no deadline expiring between consecutive visits (executions) of each task. 
    In this work we prove that $2$-\textsc{Visits} is strongly NP-complete even when the maximum multiplicity of the input is equal to $2$, settling an open question from~\cite{kVisits} and contrasting the tractability of $2$-\textsc{Visits} for simple sets. On the other hand, we prove that $2$-\textsc{Visits} is in $\mathrm{RP}$ when the number of distinct deadlines is constant, thus making progress on another open question regarding the parameterization of $2$-\textsc{Visits} by the number of numbers. We then generalize all existing positive results for $2$-\textsc{Visits} to a version of the problem where some tasks must be visited once and some other tasks twice, while providing evidence that some of these results are unlikely to transfer to $3$-\textsc{Visits}. 
    Lastly, we establish bounds for the density thresholds of $k$-\textsc{Visits}, analogous to the $(5/6)$-threshold of \textsc{Pinwheel Scheduling} [Kawamura, STOC 2024~\cite{Kawamura_5/6_stoc}]; in particular, we show a $\sqrt{2}-1/2\approx 0.9142$ lower bound for the density threshold of $2$-\textsc{Visits} and prove that the density threshold of \kV\ approaches $5/6\approx 0.8333$ for $k \to \infty$.
\end{abstract}

\newpage

\section{Introduction}

\PWS\ is a fundamental problem in scheduling theory, introduced by Holte, Mok, Rosier, Tulchinsky and Varvel~\cite{Holte_Pinwheel} in 1989. Given the deadlines of $n$ tasks, the problem asks whether it is possible to perpetually execute one task at a time, in a way such that no deadline expires between consecutive executions of the same task. Since its introduction, this problem has been extensively studied; for instance, \cite{Holte_2_distinct} and~\cite{Lin_3distinct} study the problem for constant amounts of distinct deadlines, \cite{Jacobs_Window_Scheduling_Complexity} studies its complexity and \cite{Bamboo_first} introduces an optimization version known as \emph{Bamboo Garden Trimming}.
\PWS\ has received significant attention in recent years, in particular after Kawamura's breakthrough paper~\cite{Kawamura_5/6_stoc} in STOC 2024, in which a computer-assisted proof of the \emph{density threshold conjecture} (i.e., that every \PWS\ instance with sum of inverse deadlines at most $5/6$ is a yes-instance) was given. This resolved a question that had remained open for over three decades~\cite{Chan_conjecture} despite a substantial line of work in that direction~\cite{Bar-Noy_0.6,Chan_0.7,Fishburn_density,Gasieniec_towards_5/6}.

Arguably the most important open question regarding \PWS\ is its complexity: it is known that the problem is in PSPACE ever since its introduction~\cite{Holte_Pinwheel}, but it remains open whether it is PSPACE-complete~\cite{Bosman_Replenishment}. In fact, even strong NP-hardness and membership in NP are both open questions for the problem, while weak NP-hardness was only very recently established by Kleinberg and Mishra~\cite{kleinberg_mishra_NP_hardness}, after remaining open for over three decades. Membership in NP is known only for \emph{dense} instances~\cite{Holte_Pinwheel}, while strong NP-hardness is known for a version with durations~\cite{Feinberg_Generalized_Pinwheel} and PSPACE-completeness is known for a weighted generalization~\cite{PSPACE_UAV}.  
Despite its open complexity, \PWS\ has been used to classify the complexity of other problems, in particular inventory routing problems~\cite{inventory_routing}, adding further intrigue to the aforementioned open questions.

In a SODA 2026 paper, Kanellopoulos, Pergaminelis, Kokkou, Markou and Pagourtzis~\cite{kVisits} introduced the \kV\ problem as a finite version of \PWS\ and proved that it is strongly NP-complete for $k=2$, contrasting the current status of the infinite version. As a corollary, \PWS\ becomes strongly NP-hard if each task's deadline is allowed to change once after a given amount of executions. 
This seems like substantial progress towards settling the complexity of \PWS, especially when one considers that PSPACE-completeness proofs for periodic problems often rely on the NP-hardness of finite versions (see e.g.,~\cite{Los_Alamos_periodic_PSPACE,Papadimitriou_book}). Moreover, since the existence of an infinite schedule is equivalent to the existence of an infinite schedule with a finite period~\cite{Holte_Pinwheel}, it might be possible to transfer some NP-hardness result from \kV\ to \PWS.
Even disregarding the potential of answering open questions of the infinite version, \kV\ seems like an interesting problem in its own merit: it may be more practical than \PWS\ in applications that require finite repetitions, and it may be tractable in cases in which \PWS\ is not.

With the aforementioned motivations in mind, we build upon the work of~\cite{kVisits}, establishing stronger hardness results and algorithms for \twoV, as well as bounds for the density thresholds of \kV. See Section~\ref{subsec:contributions} for a comprehensive overview of our results and their merits.

\subsection{Related Work}\label{subsec:related_work}

The \kV\ problem was introduced by Kanellopoulos et al.~\cite{kVisits} in SODA 2026 as a finite version of \PWS. Interest in \PWS\ and its variants has been reignited in recent years after Gasieniec, Smith and Wild~\cite{Gasieniec_towards_5/6} made progress towards the density threshold conjecture, and even more so after Kawamura~\cite{Kawamura_5/6_stoc} definitively proved the conjecture. This conjecture was stated by Chan and Chin~\cite{Chan_conjecture} and had remained open for over three decades.

Kanellopoulos et al.~\cite{kVisits} prove that \twoV\ is strongly NP-complete through a chain of reductions starting from \RNTDM\ (\RNTDMshort)~\cite{Flow_shop_Yu}. On the other hand, it can be solved in linear time when the input is a simple set or when the input consists of at most two distinct numbers. This leads to two natural open questions stated in~\cite{kVisits}: whether \twoV\ can be parameterized by the \emph{maximum multiplicity} or by the \emph{number of numbers}. A natural way to tackle the latter would be to try a reduction to \ILP, which is a standard method for parameterizing numerical matching variants by the number of numbers (see the paper by Fellows, Gaspers and Rosamond~\cite{Number_of_numbers} for more details). However, this method is insufficient for \twoV, because its corresponding numerical matching variant (\PM\ - Def.~\ref{def:PM}) contains two simple sets by definition, and the parameterization proposed by~\cite{Number_of_numbers} for numerical matching variants uses the combined number of numbers of \emph{all} input sets as parameter (which is $\bigO(n)$ for \PM).

Regarding density thresholds for \kV, the $(5/6)$-threshold of \PWS~\cite{Kawamura_5/6_stoc} implies that all instances with density bounded by $5/6$ admit a \kV\ schedule for all $k\in\mathbb{N}$ (cf.~\cite{kVisits}). However, as we prove, this threshold is not tight for all $k$.

Recent advances in \PWS\ and related problems include the following:  Kawamura, Kobayashi and Kusano~\cite{Kawamura_pinwheel_cover} study \emph{Pinwheel Covering}, where each task has to be executed \emph{at most} once in a specified number of time units; Kanellopoulos~\cite{finite_covering} studies a finite variant of Pinwheel Covering; Kobayashi and Lin~\cite{Pinwheel_ISAAC_2025} prove that \PWS\ cannot be solved in polynomial time under a standard complexity
assumption;\footnote{Note that this result does not imply NP-hardness for \PWS.} Biktairov et al.~\cite{Biktairov_Polyamorous} study a generalization of Bamboo Garden Trimming (BGT); Mishra~\cite{Patrolling_SODA_2026} gives an optimal density bound for Pinwheel Covering and a constant approximation for BGT; the same density bound was also concurrently proven by Kawamura and Kobayashi~\cite{Kawamura_Kobayashi_density_covering}; Mendoza-Cadena, Merino, Nielsen and Schewior~\cite{Schewior_Combinatorial_Perpetual_Scheduling_ICALP} study a combinatorial variant of BGT. Other recent progress for BGT includes~\cite{Bamboo_second,Bamboo_approx_2,DBLP:journals/jcss/GasieniecJKLLMR24,Bamboo_approx_1,Kawamura_5/6_stoc}. Recent research has also focused on generalizations of \PWS\ with task durations~\cite{SOFSEM_Kusano_pinwheel_durations} and real periods~\cite{SOFSEM_real_periods}. In a recent breakthrough paper, Kleinberg and Mishra~\cite{kleinberg_mishra_NP_hardness} showed weak NP-hardness for \PWS. Regardless, strong NP-hardness and PSPACE-completeness remain open.

Other relevant problems include \emph{Perpetual Exploration}~\cite{blin2010exclusive} and \emph{Patrolling}~\cite{Czyzowicz_Patrolling}, in particular \emph{Patrolling with Unbalanced Frequencies}~\cite{puf2018}. However, these problems usually consider special graph topologies, such as lines~\cite{puf2018}, and multiple agents~\cite{damaschke2020two}, while the problems we study correspond to complete graphs with a single agent.

\subsection{Our Contributions}\label{subsec:contributions}

The main objectives of this work are the two open questions from~\cite{kVisits} about the existence of FPT algorithms for \twoV\ parameterized by the \emph{maximum multiplicity} and by the \emph{number of numbers}; we answer the former in the negative and make progress for the latter towards the positive. Another focus is the study of density thresholds for \kV, inspired by the $(5/6)$-threshold of \PWS~\cite{Kawamura_5/6_stoc}.
Our main technical contributions include:
\begin{enumerate}
    \item A reduction chain from \NMTS\ to \twoV\ using at most two copies of each number, proving strong NP-completeness for the latter even for maximum multiplicity $2$ (Section~\ref{sec:max_mult_hardness}). This rules out FPT or even XP algorithms parameterized by the maximum multiplicity (unless $\mathrm{P}=\mathrm{NP}$), resolving an open question from~\cite{kVisits}. This result is quite surprising, as \twoV\ admits a linear-time algorithm for simple sets~\cite{kVisits}, i.e., when the maximum multiplicity is $1$.
    
    \item A reduction chain from \twoV\ to \EM, proving that \twoV\ is in the complexity class RP when the number of distinct deadlines is constant (Section~\ref{sec:number_of_numbers}). Since it is rare for a problem to be in $\mathrm{RP}$ but not in $\mathrm{P}$, we provide strong evidence towards the positive of the open question in~\cite{kVisits} regarding the parameterization by the number of numbers.
    
    \item A lower bound of $\sqrt{2}-1/2 \approx 0.9142$ for the density threshold of \twoV\ (Section~\ref{subsec:density_twoV}), significantly larger than the $(5/6)$-threshold of \PWS~\cite{Kawamura_5/6_stoc}. We complement this with a proof that the density threshold of \kV\ approaches $5/6$ when $k \to \infty$,    
    as well as other minor results regarding the density thresholds of \kV\ (see Table~\ref{table:density}).
\end{enumerate}

A detailed discussion of all of our results follows, organized per section. Section~\ref{sec:prelims} includes definitions and an overview of all past results necessary for the reader to understand this paper.


\begin{figure}[ht]
    \centering
    \resizebox{\textwidth}{!}{
\begin{tikzpicture}[
  >=latex,
  box/.style={
    draw,
    thick,
    rectangle,
    minimum height=13mm,
    minimum width=28mm,
    align=center,
    font=\bfseries\large
  },
  lab/.style={ inner sep=1pt}
]

\node[box] (p3) {3-Partition};
\node[box, right=15mm of p3] (rn) {RN3DM};
\node[box, right=15mm of rn] (in) {IN3DM};
\node[box, right=15mm of in, minimum width=28mm] (pm) {Position\\Matching};
\node[box, right=15mm of pm] (tv) {2-Visits};

\draw[->, thick, line width=1.1pt] (p3) -- node[lab, above] {\cite{Flow_shop_Yu}} (rn);
\draw[->, thick, line width=1.1pt] (rn) -- node[lab, above] {\cite{kVisits}} (in);
\draw[->, thick, line width=1.1pt] (in) -- node[lab, above] {\cite{kVisits}} (pm);
\draw[->, thick, line width=1.1pt] (pm) -- node[lab, above] {\cite{kVisits} \textbf{(*)}} (tv);

\node[box, below=18mm of p3, minimum width=36mm] (plsc)
  {Partial Latin\\Square\\Completion};

\node[box, right=11mm of plsc, minimum width=28mm] (nmts)
  {NMTS\\(distinct)};

\node[box, right=15mm of nmts, minimum width=28mm] (srnmts)
  {SRNMTS};

\node[box, right=15mm of srnmts, minimum width=28mm] (in2)
  {IN3DM\\(distinct)};

\draw[->, thick, line width=1.1pt] (plsc) -- node[lab, above] {\cite{NMTS_distinct}} (nmts);
\draw[->, thick, line width=2.5pt] (nmts) -- node[lab, above] {Thm.~\ref{thrm:srnmts_hardness}} (srnmts);
\draw[->, thick, line width=2.5pt] (srnmts) -- node[lab, above] {Thm.~\ref{thrm:IN3DM_distinct_hardness}} (in2);

\draw[->, thick, line width=2.5pt] (in2.north) -- node[lab, right] {Thm.~\ref{thrm:PM_reduction}} (pm.south);

\end{tikzpicture}}
    \caption{A map of the chain of reductions to \twoV\ in~\cite{kVisits} compared to our multiplicity-preserving chain of reductions in Section~\ref{sec:max_mult_hardness} (bold arrows). The reduction marked with a~(*) preserves the maximum multiplicity and is thus used as the final step for our result.}
    \label{fig:reductions_multiplicity}
\end{figure}
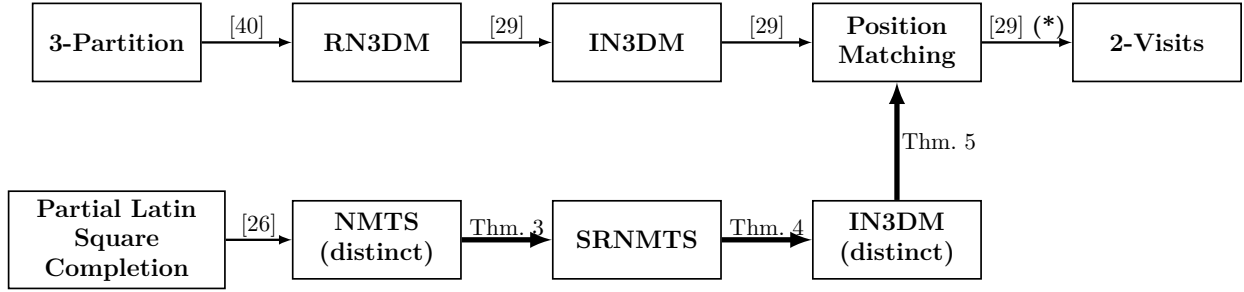

In Section~\ref{sec:max_mult_hardness} we prove that \twoV\ remains strongly NP-hard even when the maximum multiplicity of the input is $2$. Our result mostly relies on an alternative reduction from \INTDM\ (Def.~\ref{def:IN3DM}) to \PM\ (Def.~\ref{def:PM}) that only pads the input with up to $2$ copies of each number, in contrast to the reduction of~\cite{kVisits}, which pads the input with $\bigO(n)$ copies of some large number. However, this alone is insufficient to prove that \twoV\ is hard for multiplicities bounded by $2$, because the problem from which the reduction chain in~\cite{kVisits} starts (\RNTDMshort) is not known to be NP-hard for simple sets. On the contrary, the NP-hardness proof of \RNTDMshort\ by Yu, Hoogeveen and Lenstra~\cite{Flow_shop_Yu} relies heavily on padding the input with a large amount of duplicate numbers, and modifying it to forgo this seems particularly challenging. Our second contribution in Section~\ref{sec:max_mult_hardness} is thus to start another reduction chain from \NMTS\ (\NMTSshort), which is known to be strongly NP-hard even for distinct inputs by Hulett, Will and Woeginger~\cite{NMTS_distinct}. As an intermediate step, we prove that \NMTSshort\ remains hard even when the input consists of distinct numbers \emph{and one of three sets is fixed}; the latter is crucial for the reduction to \twoV.
See Figure~\ref{fig:reductions_multiplicity} for a comparison between our reduction chain and that of~\cite{kVisits}.

\begin{figure}[ht]
    \centering
    \resizebox{\textwidth}{!}{
\begin{tikzpicture}[
  >=latex,
  box/.style={
    draw,
    thick,
    rectangle,
    minimum height=13mm,
    minimum width=26mm,
    align=center,
    font=\bfseries
  },
  lab/.style={ inner sep=1pt}
]

\node[box] (2v) {2-Visits};
\node[box, right=8mm of 2v] (p3) {Position\\Matching};
\node[box, right=22mm of p3] (rn)
{EWPM\\ \small{poly weights} \\\small{multigraph}};
\node[box, right=15mm of rn] (in) {EWPM\\\small{poly weights}};
\node[box, right=12mm of in, minimum width=28mm] (pm) {Exact\\Matching};

\draw[->, thick, line width=1.1pt] (2v) -- node[lab, above] {\cite{kVisits}} (p3);
\draw[->, thick, line width=1.1pt] (p3) -- node[lab, above] {Thm.~\ref{thrm:reduction_PM_to_EWPM} \textbf{(*)}} (rn);
\draw[->, thick, line width=1.1pt] (rn) -- node[lab, above] {Lem.~\ref{lem:multigraph_to_simple}} (in);
\draw[->, thick, line width=1.1pt] (in) -- node[lab, above] {\cite{Maalouly_stacs_2022,EWPM_to_EM_reduction}} (pm);

\end{tikzpicture}}
    \caption{A map of our reductions in Section~\ref{sec:number_of_numbers}, transforming \twoV\ into \EM\ when there is a constant amount of distinct numbers. The reduction marked with a~(*) results in polynomial weights if the number of numbers is constant.}
    \label{fig:reductions_number_of_numbers}
\end{figure}
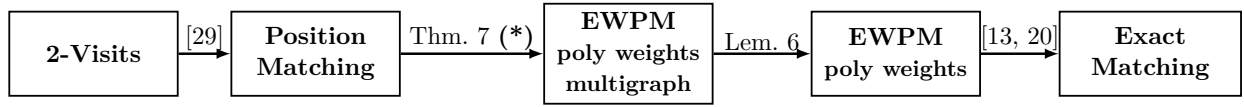

In Section~\ref{sec:number_of_numbers} we make progress on the parameterization of \twoV\ by the number of numbers, proving that it admits a randomized polynomial-time algorithm when the number of numbers is constant.\footnote{In fact, only the maximum number of numbers corresponding to the same \emph{cluster} (Def.~\ref{def:cluster}) has to be bounded by a constant. If this is the case, then the result holds even for unbounded number of numbers of the whole input.} We achieve this through a reduction to \EWPM\ (\EWPMshort), which in turn reduces to \EM~\cite{Maalouly_stacs_2022,EWPM_to_EM_reduction,Papadimitriou_EM}. The latter is known to admit a randomized polynomial-time algorithm~\cite{Vazirani_EM}, although its derandomization is a notorious open question~\cite{Maalouly_stacs_2022,Maalouly_esa_2025,Gurjar_EM_Complete}.  See Figure~\ref{fig:reductions_number_of_numbers} for an overview of our reductions in Section~\ref{sec:number_of_numbers}.

\newcommand{\lightrule}{\arrayrulecolor{black!30}\hline\arrayrulecolor{black}}
{
\renewcommand{\arraystretch}{1.15}
\begin{table}[ht]
\centering
\begin{tabular}{|l|c!{\color{black!30}\vrule}c|}
\hline
 & \textbf{Lower Density Threshold} & \textbf{Upper Density Threshold} \\
\hline
\textbf{Pinwheel}      & $5/6$~\cite{Kawamura_5/6_stoc} & $1$~\cite{Holte_Pinwheel}\\
\lightrule
\textbf{1-Visit} & $1$ (Thm.~\ref{thrm:lower_density_threshold_oneVisit}) & \multirow{4}{*}{Does not exist (Thm.~\ref{thrm:upper_density_threshold})} \\
\arrayrulecolor{black!30}\cline{1-2}\arrayrulecolor{black}
\textbf{2-Visits} & $[\sqrt{2}-\tfrac{1}{2},\ 1]$ (Thm.~\ref{thrm:lower_density_threshold_twoVisits}) &  \\
\arrayrulecolor{black!30}\cline{1-2}\arrayrulecolor{black}

\multirow{2}{*}{\textbf{k-Visits}} & $[5/6,\ 5/6 + 1/\lfloor (k-1)/6 \rfloor]$ (Thm.~\ref{thrm:lower_density_threshold_largeK}) &  \\
\arrayrulecolor{black!30}\cline{2-2}\arrayrulecolor{black}

& Approaches $5/6$ for $k\rightarrow \infty$ (Cor.~\ref{cor:density_largeK})
  & \\

\hline
\end{tabular}
\caption{Our bounds for the density thresholds of \kV\ in Section~\ref{sec:density}, compared to those of \PWS. Note that the upper bound of $5/6 + 1/\lfloor (k-1)/6 \rfloor$ is only useful for $k> 42$, since Thm.~\ref{thrm:lower_density_threshold_oneVisit} implies that the lower density threshold of \kV\ is at most $1$ for all~$k$.}
\label{table:density}
\end{table}
}

In Section~\ref{sec:density} we study the \emph{density thresholds} of \kV\ (see the opening paragraph of Section~\ref{sec:density} for formal definitions). We first prove that \oneV\ has a tight lower density threshold of~$1$, implying that the lower density threshold of \kV\ is not larger than $1$ for any $k\in\mathbb{N}$. Our main technical contribution, however, is proving that all instances with density at most $\sqrt{2}-1/2\approx 0.9142$ admit a \twoV\ schedule, which is additionally computable in linear time. The proof for this result relies on constructing the density-minimizing instance that violates a property guaranteeing the existence of a \twoV\ schedule. We additionally provide an upper bound of $5/6 + 1/\lfloor (k-1)/6 \rfloor$ for the lower density threshold of \kV, approaching $5/6$ when $k\rightarrow \infty$, thus establishing another connection between \kV\ and \PWS. Lastly, we prove that no upper density threshold exists for the former, for any $k\in \mathbb{N}$. See Table~\ref{table:density} for a summary of our results in this section compared to the known thresholds for \PWS. 

{
\begin{table}[ht]
\centering
\begin{tabular}{|c!{\color{black!30}\vrule}c|}
\hline
\textbf{(1 or 2)-Visits algorithm} & \textbf{Condition} \\
\hline
$\bigO(m+n!)$ time & - \\
\lightrule
$\bigO(m+n)$ time  & Deadlines requiring two visits are distinct \\
\lightrule
$\bigO(m+n)$ time  & At most two distinct deadlines per cluster \\
\lightrule
$\bigO(m+n\cdot c!)$ time & \multirow{2}{*}{Maximum cluster size bounded by $c$} \\
(FPT by maximum cluster size) & \\
\lightrule
Randomized polynomial-time & Constant number of distinct deadlines per cluster \\
\hline
\end{tabular}
\caption{An overview of the algorithms for \oneortwoV\ in Section~\ref{sec:one_or_two_visits} (Theorem~\ref{thrm:one_or_twoV_algos}), assuming $m$ tasks require one visit and $n$ tasks require two visits. All algorithms coincide with the respective algorithms for \twoV\ for $m=0$. Note that the naive brute force algorithm for \oneortwoV\ runs in time $\bigO((m+2n)!)$.}
\label{table:one_or_twoV}
\end{table}
}

In Section~\ref{sec:one_or_two_visits} we generalize all existing positive results for \twoV\ (including our own result from Section~\ref{sec:number_of_numbers}) to \oneortwoV, i.e., a version where $m$ tasks require one execution and $n$ tasks require two. See Table~\ref{table:one_or_twoV} for a list of the algorithms for \oneortwoV.

In Section~\ref{sec:3v_counterexample} we provide a \threeV\ counterexample that violates two of the most important properties of \twoV, with one of these violations indicating that \threeV\ may be strongly NP-complete even for distinct deadlines (which was already a conjecture in~\cite{kVisits}). Our result indicates that all algorithms in Table~\ref{table:one_or_twoV} are unlikely to generalize to versions with three or more visits per task with the currently existing tools.

\subsubsection*{Prior Work}

This work is an extension of the paper with the same title presented in ICALP 2026~\cite{ICALP_kVisits}. Compared to the conference version, this paper contains detailed proofs for all of our results, along with many minor revisions and more up-to-date related work.

\section{Preliminaries}\label{sec:prelims}

Throughout the paper, we use the notations $[n]=\{1,\dots, n\}$ and $[m,n]=\{m,\dots, n\}$ for $m,n\in \mathbb{N}$ with $m\leq n$.
Membership in NP is trivial to prove for all problems studied in this paper (except \PWS) and is thus always omitted. For all problems studied in this work, we assume that all input sets are given in a sorted (non-decreasing) order.

\subsection{Problem Definitions}

We first define \PWS~\cite{Holte_Pinwheel} and its finite version, \kV~\cite{kVisits}.

\begin{definition}[\PWS]
    Given a (multi)set of positive integers (deadlines) $D=\{ d_1,\ldots,d_n\}$, the $\PWS$ problem asks whether there exists an infinite schedule $p_1, p_2, \ldots$ , where $p_j \in [n]$ for $j\in \mathbb{N}$, such that for all $i\in [n]$ any $d_i$ consecutive entries contain at least one occurrence of $i$.
\end{definition}

\begin{definition}[\kV]\label{def:kV}
    Given a (multi)set of positive integers (deadlines) $D=\{ d_1,\ldots,d_n\}$, the $\kV$ problem asks whether there exists a schedule of length $nk$, containing each $i \in [n]$ exactly $k$ times, with the constraint that every occurrence of $i$ is at most $d_i$ positions away from the previous one (or within the first $d_i$ positions of the schedule, if it is the first occurrence of $i$).
\end{definition}

For \twoV\ in particular, we use the following definition (as in~\cite{kVisits}), which is equivalent to Definition~\ref{def:kV} for $k=2$. For an explanation on why this definition is useful, see Section~\ref{subsubsec:old_results}.

\begin{definition}[\twoV]\label{def:2V}
    Given a (multi)set of positive integers (deadlines) $D=\{ d_1,\ldots,d_n\}$, the $\twoV$ problem asks whether there exists a schedule of length $2n$, containing a \emph{primary} and a \emph{secondary} visit for each $i \in [n]$. For every $i \in [n]$, its primary visit must be at most $d_i$ positions away from the beginning of the schedule and its secondary visit must be either before its primary visit or at most $d_i$ positions after its primary visit.
\end{definition}


\begin{remark}\label{remark:2V_deadline_bound}
    We always assume that for a \twoV\ instance $D=\{d_1,\ldots,d_n\}$ it holds that $d_i \leq 2n$, $\forall i\in [n]$. Deadlines larger than $2n$ would never expire and can thus have both of their visits placed at the end of the schedule and be removed from the input.
\end{remark}

We now define various numerical matching variants that will be needed throughout the paper. 

\begin{definition}[$\NMTSshort$]\label{def:NMTS}
    Given three (multi)sets $A=\{a_1,\ldots,a_n\}$, $B=\{b_1,\ldots,b_n\}$ and $T=\{t_1,\ldots,t_n\}$ of positive integers, the \NMTS\ $(\NMTSshort)$ problem asks whether there is a subset $M$ of $A \times B \times T$ s.t. every $a_i \in A$, $b_i \in B$, $t_i \in T$ occurs exactly once in~$M$ and for every triplet $(a,b,t)\in M$ it holds that $a+b = t$.
\end{definition}

\begin{theorem}[Hulett et al. 2008~\cite{NMTS_distinct}]\label{thrm:nmts_distinct_hardness}
    \NMTSshort\ is strongly NP-complete even when all $3n$ input elements are distinct.
\end{theorem}


The following problem was defined in~\cite{kVisits} as an intermediate step for proving the strong NP-hardness of \twoV. Essentially, it is a variant of \NMTSshort\ in which an inequality is used for satisfying the targets, instead of an equality; additionally, one of the three sets is fixed and is not actually part of the input.

\begin{definition}[\INTDMshort]\label{def:IN3DM}
    Given two (multi)sets $A=\{a_1,\ldots,a_n\}$ and $T=\{t_1,\ldots,t_n\}$ of positive integers, the \INTDM\ $(\INTDMshort)$ problem asks whether there is a subset $M$ of $A \times [n] \times T$ s.t. every $a_i \in A$, $b \in [n]$, $t_i \in T$ occurs exactly once in~$M$ and for every triplet $(a,b,t)\in M$ it holds that $a+b \geq t$.
\end{definition}

We will now define \PM\ (introduced in~\cite{kVisits}), which is arguably the most crucial component for the NP-hardness proof of \twoV, as well as its algorithms. This problem is, in a certain sense, equivalent to \twoV\ (see Section~\ref{subsubsec:old_results} for details). Much of the technical challenge surrounding \twoV\ stems from the fact that the second set of this problem is derived from the first one in a recursive manner, and is thus neither part of the input, nor static as it is in \INTDMshort. We first have to define the concept of \emph{discretized sequences}~\cite{kVisits}.

\begin{definition}[Discretized Sequence]\label{def:disc}
    Given a non-decreasing sequence $D=\langle d_1,\ldots,d_n\rangle$ of positive integers, we define its \emph{discretized sequence} $A=\langle a_1,\ldots,a_n\rangle$ as follows.
    \[a_i = \begin{cases}
        d_i,\ i=n\\
        \min\{a_{i+1}-1,\ d_i\},\ i<n
    \end{cases}.\]
\end{definition}

Intuitively, the discretized sequence of a sequence of deadlines consists of the latest possible positions in which the first visits can be placed in a schedule respecting the deadlines. For example, the discretized sequence of the sequence of deadlines $D=\langle 3,5,5,7,7,7,15,15,16 \rangle$ is $A=\langle 2,3,4,5,6,7,14,15,16 \rangle$. Note that the discretized sequence of a sorted set of integers can be computed in $\bigO(n)$ time. 

\begin{definition}[\PM]\label{def:PM}
    Given a (multi)set $D=\{ d_1,\ldots,d_n\}$ and a simple set $T=\{t_1,\ldots,t_n\}$ of positive integers, let $A=\langle a_1,\ldots,a_n \rangle$ be the discretized sequence of $D$. The $\PM$ problem asks whether there is a subset $M$ of $D\times A\times T$ s.t. every $d_i \in D$, $a_i \in A$, $t_i \in T$ occurs exactly once in $M$ and for every triplet $(d,a,t)\in M$ it holds that $d \geq a$ and $d+a\geq t$.
\end{definition}

\subsection{Synopsis of key previous results}

\subsubsection{The connection between \twoV\ and \PM}\label{subsubsec:old_results}

Arguably the most crucial result for \twoV\ by Kanellopoulos et al.~\cite{kVisits} is the following lemma, which transforms \twoV\ into a numerical matching variant (\PM\ - Def.~\ref{def:PM}). 

\begin{lemma}[\cite{kVisits}]\label{lem:old_primary_discon_prop}
    If a \twoV\ instance $D=\{d_1,\ldots,d_n\}$ admits a schedule, then it admits a schedule such that all $n$ primary visits are placed in a permutation of the positions contained in the discretized sequence $A=\langle a_1,\ldots,a_n\rangle$ of $D$. Equivalently, all $n$ secondary visits are placed in a permutation of the positions in $T=[2n]\setminus A$.
\end{lemma}

The proof of Lemma~\ref{lem:old_primary_discon_prop} relies on swapping arguments between primary and secondary visits, in order to place primary visits as late as possible in the schedule (i.e., in the positions of the discretized sequence). Note that these swapping arguments rely on the fact that secondary visits can be moved to an earlier position of the schedule (if one is available), without affecting its feasibility. This is the main reason why Definition~\ref{def:2V} is used for \twoV\ instead of Definition~\ref{def:kV}; this property would be violated if we used first/second visits instead of primary/secondary visits. For details, see the full version of the paper by Kanellopoulos et al.~\cite{kvisits_arxiv}.

Using Lemma~\ref{lem:old_primary_discon_prop}, one can reduce \twoV\ to \PM\ in linear time: each deadline $d\in D$ has to be matched with a position $a\in A$ for its primary visit and a position $t \in T=[2n]\setminus A$ for its secondary visit, such that $d\geq a$ and $d+a\geq t$.\footnote{The first of these two inequalities demands that the primary visit is at most $d$ positions from the start of the schedule, while the second one demands that the secondary visit is at most $d$ positions after the respective primary (or anywhere before it). Recall that these restrictions stem directly from Def.~\ref{def:2V}.}

\begin{corollary}\label{cor:old_2V_to_PM_reduction}
    \twoV\ reduces to \PM\ in $\bigO(n)$ time.
\end{corollary}

\begin{remark}
    Kanellopoulos et al.~\cite{kVisits} present a more complicated version of the aforementioned reduction, reducing \twoV\ directly to (potentially) smaller instances of \PM, which is useful for one of the positive results of that paper. In the context of this paper, the simpler reduction to \PM\ that we described above suffices.
\end{remark}

The next two lemmas motivate our results in Sections~\ref{sec:max_mult_hardness} and~\ref{sec:number_of_numbers} respectively.



\begin{lemma}[\cite{kVisits}]\label{lem:old_distinct}
    \PM\ admits a linear-time algorithm if $D$ is a simple set. Consequently, \twoV\ also admits a linear-time algorithm if $D$ is a simple set (by Corollary~\ref{cor:old_2V_to_PM_reduction}).
\end{lemma}

The proof of Lemma~\ref{lem:old_distinct} relies on the fact that the discretized sequence of a simple set is (trivially) equal to itself. This implies that all $d\in D$ have to be matched with $a\in A$ s.t. $d=a$ for such a \PM\ instance, due to the restriction $d\geq a$ for all triplets $(d,a,t)$. It remains to check if these mandatory $(d,a)$ pairs satisfy all targets $t\in T$, which can be done in time~$\bigO(n)$.


\begin{lemma}[\cite{kVisits}]\label{lem:old_two_numbers}
    \PM\ admits a linear-time algorithm if $D$ contains at most two distinct numbers.\footnote{In~\cite{kVisits}, this result is directly stated for \twoV, but it also holds for \PM. The respective proof in the full version~\cite{kvisits_arxiv} essentially shows a greedy linear-time algorithm for \PM\ with two distinct numbers and then transfers it to \twoV.} Consequently, \twoV\ also admits a linear-time algorithm if $D$ contains at most two distinct numbers (by Corollary~\ref{cor:old_2V_to_PM_reduction}).
\end{lemma}

\subsubsection{The significance of clusters}\label{subsubsec:clusters}

The following definition from~\cite{kVisits} will be useful throughout the paper.

\begin{definition}[Cluster]\label{def:cluster}
    Let $A=\langle a_1,\ldots,a_n \rangle$ be the discretized sequence of $D=\{d_1,\ldots,d_n\}$. We call a maximal subsequence of consecutive numbers in $A$ a \emph{cluster}. If $C=\langle a_i,\ldots, a_j \rangle$ is a cluster, we say that the tasks $i,\ldots,j$ and the deadlines $d_i,\ldots,d_j$ \emph{correspond} to $C$.
\end{definition}

\begin{figure}[ht]
    \centering
    \includegraphics[width=0.7\linewidth]{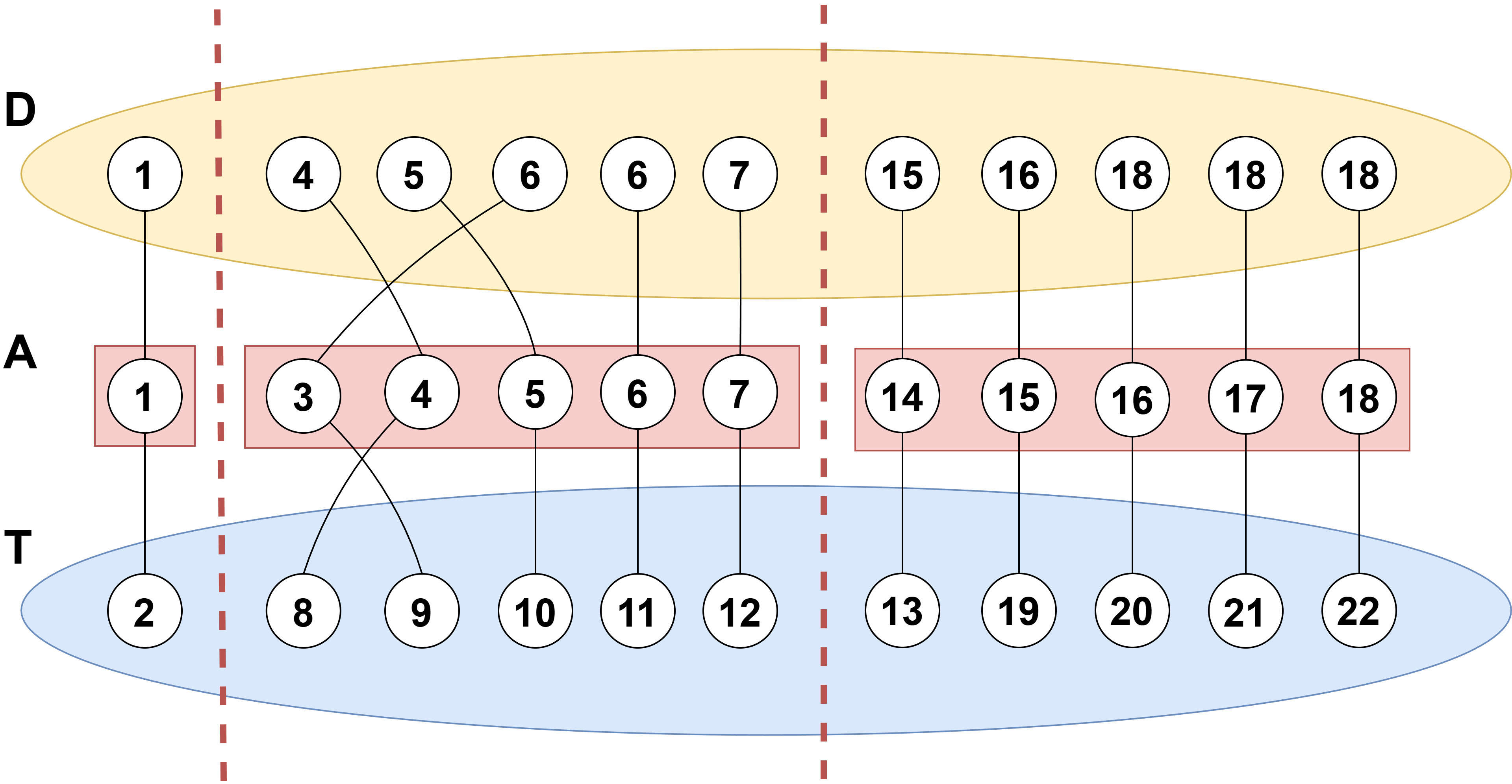}
    \caption{An example of a \PM\ instance, corresponding to the \twoV\ instance $D=\{1,4,5,6,6,7,15,16,18,18,18\}$, and its solution. $A$ is the discretized sequence of $D$ and $T=[2n]\setminus A$ (see Lemma~\ref{lem:old_primary_discon_prop}). Observe that none of the black lines dictating the solution cross the dotted lines (the borders of clusters), due to the $d\geq a$ restriction of \PM.}
    \label{fig:clusters}
\end{figure}

The next observation follows from Def.~\ref{def:disc} and is one of the reasons for the usefulness of clusters.

\begin{observation}\label{obs:cluster}
    Let $A=\langle a_1,\ldots,a_n \rangle$ be the discretized sequence of $D=\{d_1,\ldots,d_n\}$ and let $C=\langle a_i,\ldots, a_j \rangle$ be some cluster of $A$. Then, $C$ is the discretized sequence of $\{ d_i,\ldots,d_j \}$.
\end{observation}


\begin{lemma}[Self-reduction]\label{lem:PM_self_reduction}
    \PM\ reduces in time $\bigO(n)$ to solving a \PM\ instance for the numbers corresponding to each cluster.
\end{lemma}

Lemma~\ref{lem:PM_self_reduction} is heavily implied in~\cite{kVisits}, but not directly proven in the form that we desire for this work. For the sake of completeness, we include its formal proof in Appendix~\ref{appendix:PM_self_reduction}. See Figure~\ref{fig:clusters} for some intuition for the proof.

\subsubsection{Pinwheel variants and the concept of density}\label{subsubsec:density}

A quantity often used in \PWS\ and its variants (see e.g.,~\cite{Bar-Noy_0.6, Holte_Pinwheel, Kawamura_5/6_stoc, Kawamura_pinwheel_cover, Patrolling_SODA_2026}) is the \emph{density} of the input, defined as follows.

\begin{definition}[Density]\label{def:density}
    For a (multi)set $D=\{d_1,\ldots,d_n\}$, its \emph{density} is defined as $\mathrm{Dens}(D)=\sum_{i=1}^n 1/d_i$.
\end{definition}

It is straightforward to prove that every \PWS\ instance $D$ with $\mathrm{Dens}(D) > 1$ admits no schedule (see~\cite{Holte_Pinwheel} for a formal proof). On the other hand, it was recently proven that all \PWS\ instances $D$ with $\mathrm{Dens}(D) \leq 5/6$ admit a schedule~\cite{Kawamura_5/6_stoc}. Note that this density threshold is tight, due to the instance $I=\{2,3,1/\varepsilon\}$  with $\mathrm{Dens}(I)=5/6 + \varepsilon$, which admits no infinite schedule for every $\varepsilon > 0$~\cite{Chan_conjecture,Gasieniec_towards_5/6}.

The aforementioned result by~\cite{Kawamura_5/6_stoc} implies that all \kV\ instances $D$ with $\mathrm{Dens}(D) \leq 5/6$ admit a schedule, for all $k\geq 1$: it suffices to take an infinite schedule and remove all visits except the $k$ first visits of each task, preserving feasibility (cf.~\cite{kVisits}). 

\subsubsection{\EM\ and variations}\label{subsubsec:exact_matching}

We give the definitions of \EM\ and its weighted generalization, which we will use in Section~\ref{sec:number_of_numbers}.
\EM\ was introduced by Papadimitriou and Yannakakis~\cite{Papadimitriou_EM} in 1982, where they conjectured that it is NP-complete. However, Mulmuley, Vazirani and Vazirani~\cite{Vazirani_EM} showed that the problem admits a randomized polynomial-time algorithm, placing it in the complexity class~$\mathrm{RP}$. To this day, \EM\ remains notorious for being one of the few natural problems that are in $\mathrm{RP}$, but not known to be in $\mathrm{P}$.

\begin{definition}[\EM]\label{def:EM}
    Given a graph $G=(V,E)$, a subset  $E' \subseteq E$ of \emph{red} edges and a positive integer $k$, the \EM\ problem asks whether there exists a perfect matching in $G$ involving \emph{exactly} $k$ red edges.
\end{definition}

\begin{definition}[\EWPMshort]\label{def:EWPM}
    Given a weighted graph $G=(V,E,w)$ and an integer $W$, the \EWPM\ $(\EWPMshort)$ problem asks whether there exists a perfect matching~$M$ in $G$ with $\sum_{e \in M} w(e) =W$.
\end{definition}

\begin{theorem}[Gurjar et al. 2016~\cite{EWPM_to_EM_reduction}, El Maalouly 2023~\cite{Maalouly_stacs_2022}]\label{thrm:EWPM_to_EM}
    There is a polynomial-time reduction from \EWPMshort\ with polynomially-bounded weights to \EM. If the initial graph is bipartite, the resulting graph is also bipartite. If the initial \EWPMshort\ instance has $n$ vertices, $m$~edges and weights bounded by $W$, then the resulting \EM\ instance has $\bigO(n+mW)$ vertices and $\bigO(mW)$ edges. The reduction runs in time $\bigO(mW)$.
\end{theorem}

Note that \EWPMshort\ is known to be NP-complete when its weights are exponential~\cite{EWPM_to_EM_reduction}. However, when its weights are polynomially bounded, it is in $\mathrm{RP}$ by Theorem~\ref{thrm:EWPM_to_EM}.



\section{Hardness of \twoV\ relative to maximum multiplicity }\label{sec:max_mult_hardness}

In this section we prove that \twoV\ remains strongly NP-complete even when the maximum multiplicity is $2$, resolving an open question from~\cite{kVisits}. Our chain of reductions starts from \NMTS\ (\NMTSshort). See Figure~\ref{fig:reductions_multiplicity} for an overview.

The key idea for our result in this section is that padding an input sequence with two copies of certain larger numbers forces the discretized sequence of the input to consist of consecutive numbers starting from $1$, as demonstrated in the following example.

\begin{example}\label{example:discseq}
    Consider input $\langle 2,4,5,8,8,10 \rangle$. Its discretized sequence is $\langle 2,4,5,7,8,10 \rangle$. Let us pad the input with $2$ copies of some larger numbers: $\langle 2,4,5,8,8,10,11,11,12,12,13,13,14,14 \rangle$. Observe that the discretized sequence of the modified input is $\langle 1,2,3,4,5,6,7,8,9,10,11,12,13,14 \rangle=[14]$. Additionally, the maximum multiplicity of the modified input is $2$ (granted that the maximum multiplicity of the initial input was not already larger than $2$).
\end{example}

This greatly simplifies the \PM\ problem (and thus also \twoV): the second set used for the matching now becomes static and equal to $[n]$, rendering reductions from classical numerical matching problems possible, while only increasing the maximum multiplicity of the input by at most $1$. Before we utilize this idea in Section~\ref{subsec:IN3DM_to_PM}, we prove an NP-hardness result for an appropriate variant of NMTS, serving as the first building block in our reduction chain.

\subsection{Restricting \NMTS}

We build upon the work of Hulett et al.~\cite{NMTS_distinct} (Theorem~\ref{thrm:nmts_distinct_hardness}) in order to prove that \NMTSshort\ remains hard even when one set is equal to $[n]$ and each of the other two sets consists of distinct numbers. To this end, we define the following restricted version of the problem, which we will use in our chain of reductions leading to \twoV.

\begin{definition}[\SRNMshort]\label{def:SRNM}
    Given two (simple) sets of positive integers $A=\{a_1,\ldots,a_n\}$ and $T=\{t_1,\ldots,t_n\}$, the \SRNM\ $(\SRNMshort)$ problem asks whether there is a subset $M$ of $A \times [n] \times T$ s.t. every $a_i \in A$, $b \in [n]$, $t_i \in T$ occurs exactly once in~$M$ and for every triplet $(a,b,t)\in M$ it holds that $a+b = t$.
\end{definition}

We call this problem \emph{semi-restricted} in order to distinguish from the \RNTDM\ (\RNTDMshort) problem defined by Yu et al.~\cite{Flow_shop_Yu}, in which two out of three sets are fixed (instead of just one). We remark that \RNTDMshort\ is not known to be NP-hard for distinct inputs; on the contrary, the reduction from \tP\ by Yu et al.~\cite{Flow_shop_Yu} appears to rely heavily on padding the input with $\mathrm{poly}(n)$ duplicates of $0$ and some large number. Modifying this method to forgo duplicates does not seem straightforward, which is the reason why we do not use \RNTDMshort\ for our reductions in this paper.

\begin{observation}\label{obs:nmts_target_ineq}
    For every non-trivial instance of \NMTSshort, it holds that $\max(T)>\max(A)$ and $\max(T)>\max(B)$; otherwise, we have a trivial no-instance. Hence, we may assume that these inequalities hold for the following reduction from \NMTSshort. 
\end{observation}

\begin{theorem}\label{thrm:srnmts_hardness}
\SRNMshort\ is strongly NP-complete.
\end{theorem}

\begin{proof}
    We will reduce \NMTSshort\ to \SRNMshort. By Theorem~\ref{thrm:nmts_distinct_hardness}, we may assume a \NMTSshort\ instance $I=(A,B,T)$ where $A=\{a_1,\ldots,a_n\}, B=\{b_1,\ldots,b_n\}, T=\{t_1,\ldots,t_n\}$ are all simple sets (instead of multisets). Additionally, since \NMTSshort\ is \emph{strongly} NP-hard, we may assume that the values of all input elements are bounded by $\mathrm{poly}(n)$. We will modify the input in order to force the second set to consist of consecutive numbers.

    We define $S=[\max(B)] \setminus B$, i.e., the set consisting of all the positive integers up to $\max(B)$ that are not in $B$. Let $S=\{s_1,\ldots,s_m\}$. For each $i \in [m]$, we add the element $3i\max(T)-s_i$ to~$A$ and the element $3i\max(T)$ to $T$.\footnote{The constant $3$ used here is chosen with the purpose of establishing a safe margin, making the separation with the added elements more intuitive. The arguments presented also work by replacing it with $2$.} Let $A'$ and $T'$ be the respective sets obtained after adding these elements and define the \SRNMshort\ instance $I'=(A',T')$, with $|A'|=|T'|=n+m$. Recall that the second of the three sets in \SRNMshort\ is not part of the input (see Def.~\ref{def:SRNM}); for $I'$, it is the set $[\max(B)]=B\cup S$, by definition. Note that $A'$ and $T'$ are simple sets by Observation~\ref{obs:nmts_target_ineq}.

    We now prove the correctness of the reduction. It is clear that if $I$ has a solution, then $I'$ has a solution: one can simply extend a solution of~$I$ with the $m$ triplets $(3i\max(T) - s_i,\ s_i,\ 3i\max(T))$, for $i\in [m]$, thus obtaining a solution for~$I'$.
    
    It remains to prove the converse, i.e., that if $I'$ has a solution, then $I$ also has a solution. We will prove by induction that picking the aforementioned $m$ triplets is mandatory in order to solve~$I'$. We use $i=m$ as the induction basis. Let $t^*=\max(T')=3m\max(T)$. Note that $t^*$ must be matched with some $a\in A'$ and some $b\in [\max(B)]$ s.t. $a=t^*-b$. Combined with Observation~\ref{obs:nmts_target_ineq}, this implies
    \begin{equation}\label{ineq_srnm}
        a\geq t^* - \max(B)> t^* -\max(T)=(3m-1)\max(T).
    \end{equation}
    The only element in $A'$ satisfying \eqref{ineq_srnm} is $a^*=\max(A')=3m\max(T)-s_m$, hence it is mandatory to include the triplet $(a^*,\ t^*-a^*,\ t^*)=(3m\max(T) - s_m,\ s_m,\ 3m\max(T))$ in any solution of $I'$.\footnote{Observe that the second largest element in $A'$ is $(3m-1)\max(T)-s_{m-1}$, which already violates~\eqref{ineq_srnm}. Additionally, note that $s_m < \max(T)$ (cf. Observation~\ref{obs:nmts_target_ineq}), implying that $a^*$ satisfies~\eqref{ineq_srnm}.}
    
    As induction hypothesis, suppose that for some $i\in [m]$ it holds that the triplets $(3j\max(T) - s_j,\ s_j,\ 3j\max(T))$, $\forall j \in [i+1,m]$, have to be included in any solution of $I'$. Removing these matched elements from the input allows us to use exactly the same argument as in the induction basis in order to prove that $(3i\max(T) - s_i,\ s_i,\ 3i\max(T))$ also has to be included in any solution of $I'$. This concludes the inductive step, thus proving that all of the elements in $A'\setminus A$, $S$ and $T'\setminus T$ have to be matched with each other in any solution of $I'$. Hence, $I'$ has a solution only if there is a partial solution to the rest of the input, i.e., only if $I=(A,B,T)$ has a solution. This concludes the correctness of the reduction.

    The reduction described is polynomial-time, since we added $\bigO(\mathrm{poly}(n))$ elements to the input. It also preserves strong NP-hardness since all added elements have values bounded by $\mathrm{poly}(n)$.
\end{proof}

\subsection{From \SRNMshort\ to \INTDMshort}

We now use \SRNMshort\ to obtain a stronger version of the \INTDMshort\ hardness theorem of~\cite{kVisits}.

\begin{theorem}\label{thrm:IN3DM_distinct_hardness}
    \INTDMshort\ is strongly NP-complete even when $A$ and $T$ are simple sets.
\end{theorem}

\begin{proof}
    We will reduce \SRNMshort\ to \INTDMshort.
    We observe that for any non-trivial \SRNMshort\ instance $(A=\{a_1,\ldots,a_n\},T=\{t_1,\ldots,t_n\})$ it holds that
    \begin{equation}\label{eq_ineq}
        \sum_{i\in [n]}(a_i + i) = \sum_{i\in [n]}t_i.
    \end{equation}
    Otherwise, we would have a trivial no-instance. Now, assume a \SRNMshort\ instance $I=(A,T)$ satisfying~\eqref{eq_ineq} and consider the respective \INTDMshort\ instance $I'=(A,T)$. By Def.~\ref{def:SRNM},~\ref{def:IN3DM}, it is clear that any solution of $I$ is a solution of $I'$. For the converse, notice that if~\eqref{eq_ineq} holds, then $I'$ cannot have a solution that includes some triplet $(a,b,t)$ with $a+b>t$. Thus, it holds that a (perfect) matching $M\subseteq A \times [n] \times T$ s.t. $\forall (a,b,t)\in M: a+b\geq t$ also satisfies $\forall (a,b,t)\in M: a+b= t$. Hence, any solution of $I'$ is also a solution of $I$, which concludes the reduction. The theorem then follows from Theorem~\ref{thrm:srnmts_hardness}, combined with the fact that $A$ and $T$ are simple sets by Def.~\ref{def:SRNM}.
\end{proof}

\subsection{From \INTDMshort\ to \PM}\label{subsec:IN3DM_to_PM}

In this subsection we present a reduction from \INTDMshort\ to \PM\ that increases the maximum multiplicity of the first set by at most $1$, while preserving the maximum multiplicity of the target set $T$. Note that $T$ containing only distinct integers is crucial for the reduction from \PM\ to \twoV\ (cf.~\cite{kVisits}).
Additionally, recall that \PM\ can be solved in polynomial time when the input consists solely of simple sets (Lemma~\ref{lem:old_distinct}), while \INTDMshort\ is strongly NP-complete even in that case (Theorem~\ref{thrm:IN3DM_distinct_hardness}). Thus, concerning multiplicity preservation, this is the best possible reduction unless $\mathrm{P}=\mathrm{NP}$.

We first present an auxiliary lemma that bounds the values of elements in \INTDMshort.

\begin{lemma}\label{lem:IN3DM_value_bound}
    \INTDMshort\ is NP-complete, even when all the following properties hold.
    \begin{enumerate}
        \item $\min(A) \geq n$.
        \item There is some polynomial $P(n)$ such that $\max(A)\leq P(n)$ and $\max(T)\leq P(n)$.
        \item $A$ and $T$ are simple sets.
    \end{enumerate}
\end{lemma}

\begin{proof}
    By Theorem~\ref{thrm:IN3DM_distinct_hardness}, \INTDMshort\ is NP-complete even when $A$ and $T$ are simple sets and there is a polynomial $Q(n)$ s.t. $\max(A)\leq Q(n)$ and $\max(T)\leq Q(n)$ (by definition of strong NP-hardness). Take such an instance $I=(A,T)$ and increase all elements in $A$ and $T$ by $n$ to obtain another \INTDMshort\ instance $I'=(A',T')$. Observe that $I$ and $I'$ are equivalent instances by Def.~\ref{def:IN3DM}. By construction, we have $\min(A') \geq n$ and $A',T'$ are simple sets. Setting $P(n)=Q(n)+n$ completes the proof.
\end{proof}

Recall that for \PM\ the second set is the \emph{discretized sequence} (Def.~\ref{def:disc}) of the first one  and for any triplet $(a,b,t)$ in its solution it must hold that $a\geq b$ and $a+b \geq t$. Most of the technical difficulty of the following reduction stems from these restrictions. The main idea is to pad the first set with duplicate numbers in order to force its discretized sequence to consist of consecutive numbers starting from $1$ (like the second set of \INTDMshort). See Example~\ref{example:discseq} for some intuition regarding this transformation. At the same time, we utilize the $\min(A) \geq n$ inequality from Lemma~\ref{lem:IN3DM_value_bound} to force all candidate triplets $(a,b,t)$ to satisfy $a\geq b$, thus rendering that restriction irrelevant and forcing an equivalence between \INTDMshort\ and \PM.

We are now ready to present our reduction to \PM.

\begin{theorem}[Main reduction]\label{thrm:PM_reduction}
    There is a polynomial-time reduction from \INTDMshort\ to \PM, such that all elements added to the input have multiplicity at most $2$ and values polynomial in $n$.
\end{theorem}

\begin{proof}
    By Lemma~\ref{lem:IN3DM_value_bound}, we may assume for our reduction an \INTDMshort\ instance $I=(A,T)$ s.t. $\min(A) \geq n$, $\max(A)\leq P(n)$ and $\max(T)\leq P(n)$ for some polynomial $P(n)$, with $A$ and $T$ being simple sets.

    We construct a \PM\ instance $I'=(A',T')$ as follows.
    \begin{itemize}
        \item $A'$ contains the elements of $A$, along with $2$ copies of each of the following numbers: $P(n)+1,\ P(n)+2,\ldots,\ 2P(n)-n$.
        \item $T'$ contains the elements of $T$, along with the numbers in $L=[P(n)+n+2,\ 4P(n)-2n+1]$ but with every third number of $L$ being omitted, i.e., it contains: $P(n)+n+2,\ P(n)+n+3,\ P(n)+n+5,\ P(n)+n+6,\ P(n)+n+8,\ldots,\ 4P(n)-2n-3,\ 4P(n)-2n-1,\ 4P(n)-2n$. Note that $|L|=3P(n)-3n$ is a multiple of $3$, which implies $\max(L)$ is omitted from $T'$ (and so on for elements of the form $\max(L)-3i$). In total, $(2/3)\cdot |L| = 2P(n)-2n$ elements are added to $T$ to obtain~$T'$.
    \end{itemize}
    Observe that $|A'|=|T'|=2P(n)-n$ and that $T'$ is a simple set, while $A'$ has maximum multiplicity equal to $2$. We define $m=2P(n)-n$.
    
    We will now prove that the \PM\ instance $I'=(A',T')$ has a solution if and only if the \INTDMshort\ instance $I=(A,T)$ has a solution. 
    
    First, we compute the discretized sequence $S$ of $A'$, proving that it is equal to $[m]$, i.e., it consists of consecutive numbers, similar to Example~\ref{example:discseq}. By Definition~\ref{def:disc}, $s_m=\max(A')=2P(n)-n$. Since there are two copies of $2P(n)-n$, we obtain $s_{m-1}=s_m-1=2P(n)-n-1$. Then, there are two copies of $2P(n)-n-1$ in $A'$; hence $s_{m-2}=2P(n)-n-2$, $s_{m-3}=2P(n)-n-3$, and so on, up until $s_{n+1} = s_{m-{(2P(n)-2n-1)}}=2P(n)-n-(2P(n)-2n-1)=n+1$. Since $\min(A) \geq n$ (by Lemma~\ref{lem:IN3DM_value_bound}), it holds that $s_i = s_{i+1}-1$ for $i\in [n]$, meaning that $s_n = n,\ s_{n-1}=n-1\ldots, s_1=1$. We conclude that $S=\{1,\ldots,2P(n)-n\}=[m]$. Recall that this set is used as the second of the three sets in \PM\ (see Def.~\ref{def:PM}).

    We now prove that for any feasible solution of $I'$ all elements in $A'\setminus A$, $S\setminus [n]$ and $T'\setminus T$ have to be matched with each other, and that such a matching is indeed possible. Note that the three aforementioned sets consist of the $m-n$ largest elements of $A',S,T'$ respectively. Consider the maximum elements of each of the three sets, i.e., $a_m=\max(A'),\ s_m=\max(S),\ t_m=\max(T')$ and observe that $t_m=a_m+s_m = 2m$ and $a_m = s_m = m$, meaning that $(a_m,s_m,t_m)$ is a feasible \PM\ triplet. Additionally, the triplet $(a_m,s_m,t_m)$ must be included\footnote{Note that the triplet $(a_{m-1},s_m,t_m)$ is also feasible and equivalent to the one we picked, since $a_{m-1} = a_{m}$. In the context of this reduction, when we say that a triplet must be included in the solution, we imply that a triplet using other elements with the same values also suffices (without making any difference).} in any feasible solution of $I'$, since for any elements $a\in A',\ s\in S$ with $a\neq a_m$ or $s\neq s_m$ it would be $t_m > a+s$, violating one of the restrictions of \PM. Assume that we match the elements $a_m,s_m,t_m$ with each other and remove them from the input.
    
    We repeat the same process with the maximum available (i.e., unmatched) elements in $A',S,T'$. For the second step, we have $a_{m-1}=a_m$ (recall that we added two copies of each number to $A'$ during the construction), $s_{m-1}=s_m -1$ and $t_{m-1} = t_m -1$. Combining these with $t_m=a_m+s_m$ and $a_m \geq s_m$, we obtain $t_{m-1}=a_{m-1}+s_{m-1}$ and $a_{m-1} \geq s_{m-1}$, implying that $(a_{m-1},s_{m-1},t_{m-1})$ is a feasible \PM\ triplet; additionally, this triplet must be included in any solution of $I'$ for the same reason as before: $t_{m-1}$ cannot be satisfied by any available elements smaller than $a_{m-1},s_{m-1}$. For the third step, both copies of $\max(A')$ have been used up, meaning that $a_{m-2}=a_{m-1}-1$. However, since every third element in $L=[P(n)+n+2,\ 4P(n)-2n+1]$ is excluded from $T'$ (by construction), we have $t_{m-2}=t_{m-1}-2$. Hence, the triplet $(a_{m-2},s_{m-2},t_{m-2})$ is also feasible and mandatory for any solution of $I'$, for the same reasons (i.e., $t_{m-2}=a_{m-2}+s_{m-2}$, $a_{m-2}\geq s_{m-2}$ and $t_{m-2}$ cannot be matched with any elements smaller than $a_{m-2},s_{m-2}$ because it would be strictly larger than their sum). Iteratively, we can prove the same (in)equalities for all elements in $A'\setminus A$, $S\setminus [n]$ and $T'\setminus T$ with the same pattern:
    \begin{itemize}
        \item In even steps, the maximum available element of $A'$ is the same as in the previous step, while the maximum available elements of $S$ and $T'$ each decrease by $1$.
        \item In odd steps (excl. the first one), the maximum available elements of $A'$ and $S$ each decrease by $1$ and the maximum available element of $T'$ decreases by $2$, compared to the previous step.
    \end{itemize}
    In both cases, the respective triplet is feasible and has to be included in any solution, by the same (in)equalities that we proved for the second and third step respectively.

    At this point we have shown that there is a partial solution of $I'$ that matches all elements in $A'\setminus A$, $S\setminus [n]$ and $T'\setminus T$ with each other, and that this partial solution is mandatory for every solution of $I'$. We infer that $I'$ has a solution if and only if the rest of its input has a partial solution, i.e., if and only if there is a (perfect) matching $M\subseteq A\times [n] \times T$ s.t. $\forall (a,b,t)\in M$ it holds that $a\geq b$ and $a+b\geq t$. Now, observe that the first of these two inequalities is actually satisfied by every pair $(a,b)\in A\times [n]$ because $\min(A) \geq n$ (by Lemma~\ref{lem:IN3DM_value_bound}). Hence, the $a\geq b$ restriction can be ignored, meaning that $I'$ has a solution if and only if there is a (perfect) matching $M\subseteq A\times [n] \times T$ s.t. $\forall (a,b,t)\in M: a+b\geq t$. This is precisely the \INTDMshort\ instance $I=(A,T)$. We obtain that $I'$ has a solution if and only if $I$ has a solution, which concludes the correctness of the reduction.
\end{proof}

We obtain the following corollary from Theorem~\ref{thrm:PM_reduction} and Lemma~\ref{lem:IN3DM_value_bound}.

\begin{corollary}\label{cor:PM_maxmult2}
    \PM\ is strongly NP-complete even when the maximum multiplicity of $D$ is equal to $2$.
\end{corollary}


\subsection{From \PM\ to \twoV}

\begin{observation}\label{obs:PM_to_2V}
    The reduction from \PM\ to \twoV\ by Kanellopoulos et al.~\cite{kVisits} preserves the maximum multiplicity. More specifically, for a \PM\ instance $(D,T)$ it requires $T$ to be a simple set and it constructs the set of deadlines for \twoV\ by padding $D$ with distinct numbers that are either smaller than $\min(D)$ or larger than $\max(D)$. 
\end{observation}

We obtain the main result of this section by combining Corollary~\ref{cor:PM_maxmult2} with Observation~\ref{obs:PM_to_2V}.

\begin{theorem}\label{thrm:main}
    \twoV\ is strongly NP-complete even when the maximum multiplicity of the input is equal to $2$.
\end{theorem}

We thus obtain the following corollary, settling an open question stated by~\cite{kVisits} in SODA 2026 and contrasting the tractability of \twoV\ for simple sets (Lemma~\ref{lem:old_distinct}).

\begin{corollary}\label{cor:open_question}
    There is no FPT or XP algorithm for \twoV\ parameterized by the maximum multiplicity, unless $\mathrm{P}=\mathrm{NP}$.
\end{corollary}

\section{\twoV\ parameterized by the number of distinct numbers}\label{sec:number_of_numbers}

In this section we prove that \PM\ is in $\mathrm{RP}$ when the number of distinct numbers in~$D$ (or simply \emph{number of numbers}) is constant, by reducing it to \EWPM\ (\EWPMshort--Def.~\ref{def:EWPM}). This naturally transfers to \twoV\ through Corollary~\ref{cor:old_2V_to_PM_reduction}.

Let $p$ be the number of numbers of the input set $D$ of \PM. The main idea of the following reduction is to model \PM\ as a bipartite matching instance between $A$ and $T$, while modeling $D$ as edge weights. Bounding $p$ allows these weights to be polynomial, which is important for \EWPMshort\ (see Section~\ref{subsubsec:exact_matching}). We remark that it only makes sense to bound the number of numbers of $D$, since $A$ and $T$ both consist of distinct numbers by definition.

\begin{theorem}\label{thrm:reduction_PM_to_EWPM}
    There is a reduction from \PM\ to \EWPMshort\ in multigraphs with $\bigO(n^2p)$ edges and weights bounded by $\bigO(n^{p-1})$, running in time $\bigO(n^2p)$.
\end{theorem}

\begin{proof}
    Let $I=(D,T)$ be a \PM\ instance with $|D|=|T|=n$ and $D$ containing $p$ distinct numbers $d_1,\ldots,d_p$. For $i\in [p]$, we define $n_i$ as the number of copies of $d_i$ in $D$. Note that $n=\sum_{i=1}^p n_i$. Let $A=\langle a_1,\ldots,a_n \rangle$ be the discretized sequence of $D$ and recall that $T=\{t_1,\ldots,t_n\}$ is a simple set by definition.

    We construct an \EWPMshort\ instance $J=(G,W)$ as follows.
    Create a weighted bipartite multigraph $G=(U\cup V,E,w)$ where 
    $U$ corresponds to the elements of $A$ (vertex $u_i$ corresponds to~$a_i$) and $V$ corresponds to elements of $T$ (vertex $v_i$ corresponds to $t_i$).
    For each $i\in[p]$, we define $w_i = (n+1)^{\,i-1}$. For every triplet $(i,j,k)\in [p]\times [n]\times [n]$, if $d_i \ge a_j$  and $d_i+a_j \ge t_k$, we add an edge $e=(u_j,v_k)$ with weight $w(e)=w_i$. Note that parallel edges might exist if multiple indices~$i$ satisfy the same inequalities, hence $G$ is a multigraph. Lastly, we define the target weight for \EWPMshort\ as $W=\sum_{i=1}^p n_i\, w_i$. The reduction described runs in time $\bigO(n^2p)$.

    We now prove that $I$ has a solution iff $J$ has a solution.

    ($\Rightarrow$) Suppose $I$ has a feasible solution
    $M \subseteq D\times A \times T$. For each triplet $(d,a_j,t_k) \in M$, let $d=d_i$ be its value among $\{d_1, \dots,d_p\}$. Since $M$ is a feasible solution, we obtain $d_i\ge a_j$ and $d_i+a_j \ge t_k$, hence by construction there is an edge between $u_j$ and $v_k$ with weight $w_i$. Selecting that edge for every triplet in $M$ yields a perfect matching $M'$ in $G$ (because every $a\in A$ and $t\in T$ appears exactly once in $M$ by Def.~\ref{def:PM}). Moreover, since $d_i$ appears exactly~$n_i$ times in $D$, exactly~$n_i$ chosen edges have weight $w_i$, thus $\sum_{e\in M'} w(e)=\sum_{i=1}^p n_i w_i=W$. Therefore, $J$ is a yes-instance of \EWPMshort.
    
    ($\Leftarrow$) Assume $G$ has a perfect matching $M'$ with total weight $W=\sum_{i=1}^p n_i\, w_i$ $(0\leq n_i \leq n)$. For $i\in[p]$, let $c_i$ be the number of edges of weight $w_i=(n+1)^{i-1}$ in $M'$ and observe that $0 \leq c_i \leq n,\ \forall i\in [p]$, because $M'$ is a perfect matching. The total weight of $M'$ is $$W=\sum_{i=1}^p c_i\, (n+1)^{i-1}=\sum_{i=1}^p n_i\, (n+1)^{i-1}.$$    
    Due to the uniqueness of base $(n+1)$ representation, there is a unique way to build the number
    $W$ using at most $n$ copies of each number $(n+1)^{i-1}$. Hence, it must be
    $c_i=n_i$, $\forall i\in [p]$. 
    We will now construct a solution for the \PM\ instance $I$. For each edge $(u_j,v_k) \in M'$ with weight $w_i$, we add to our solution the triplet $(d_i,a_j,t_k)$. An edge exists iff $d_i\ge a_j$ and $d_i + a_j \ge t_k$ by construction, hence every chosen triplet is feasible. Since $M'$ is perfect, every $a_j\in A$ and every $t_k\in T$ appears exactly once. Finally, $d_i$ appears exactly $c_i=n_i$ times, i.e., we use exactly the multiset $D$.
    Thus, we obtain a solution for~$I$.
\end{proof}

To the best of our knowledge, \EM\ and \EWPMshort\ have not been considered in multigraphs, although it is likely that some of their results generalize to multigraphs easily. We thus include Lemma~\ref{lem:multigraph_to_simple} for the sake of completeness. Its proof is a standard transformation of a multigraph to a simple graph by replacing edges with paths (see Figure~\ref{fig:EWPM}). Note that for the purposes of this work we only need the following reduction for \emph{bipartite} multigraphs, although it holds even for general multigraphs.

\begin{figure}[ht]
     \centering
     \begin{subfigure}{0.45\textwidth}
         \centering
         \includegraphics[width=\linewidth]{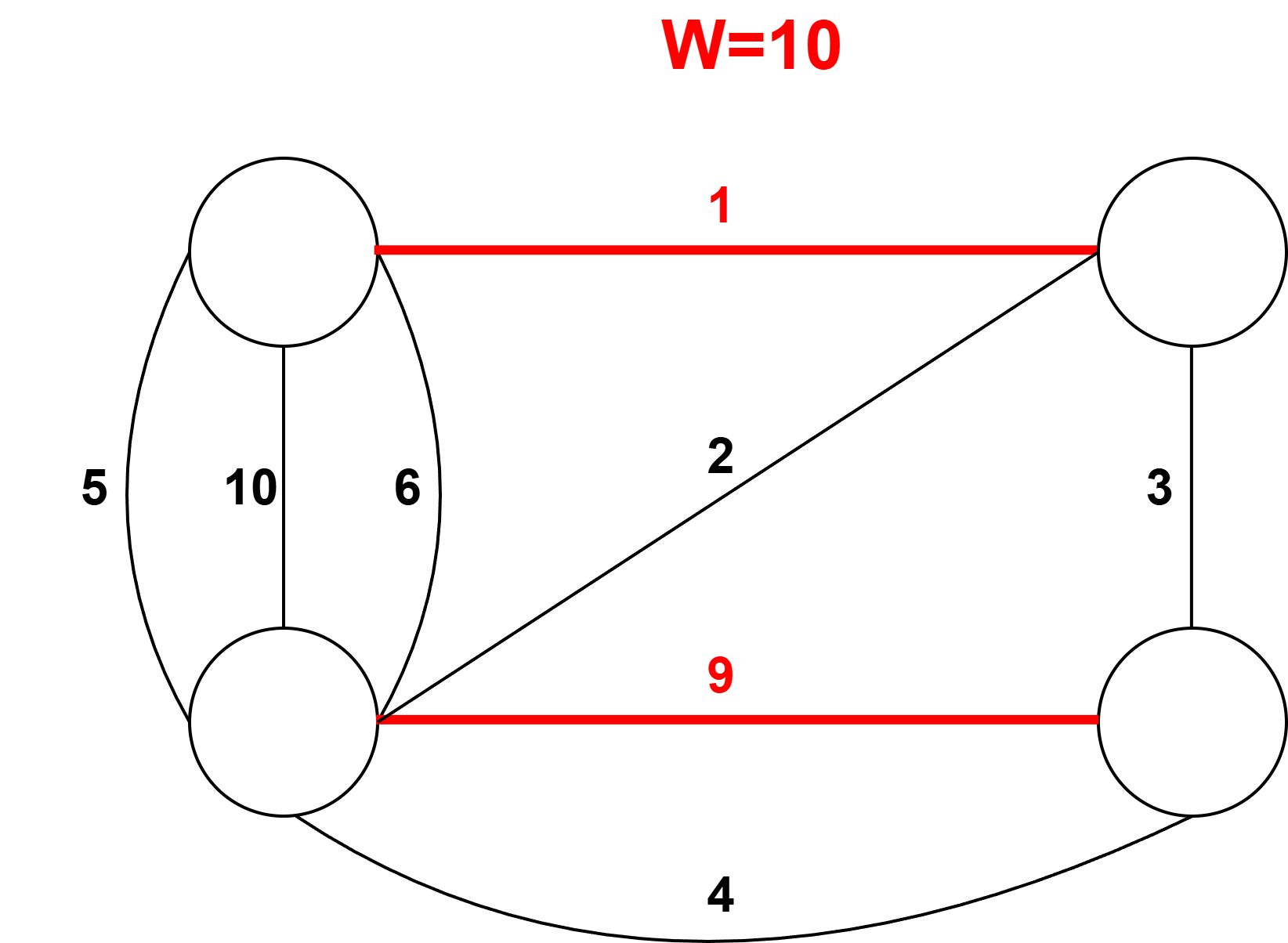}
         \vspace{20pt}
         \label{fig:1a}
     \end{subfigure}\hfill
     \begin{subfigure}{0.5\textwidth}
         \centering
         \includegraphics[width=\linewidth]{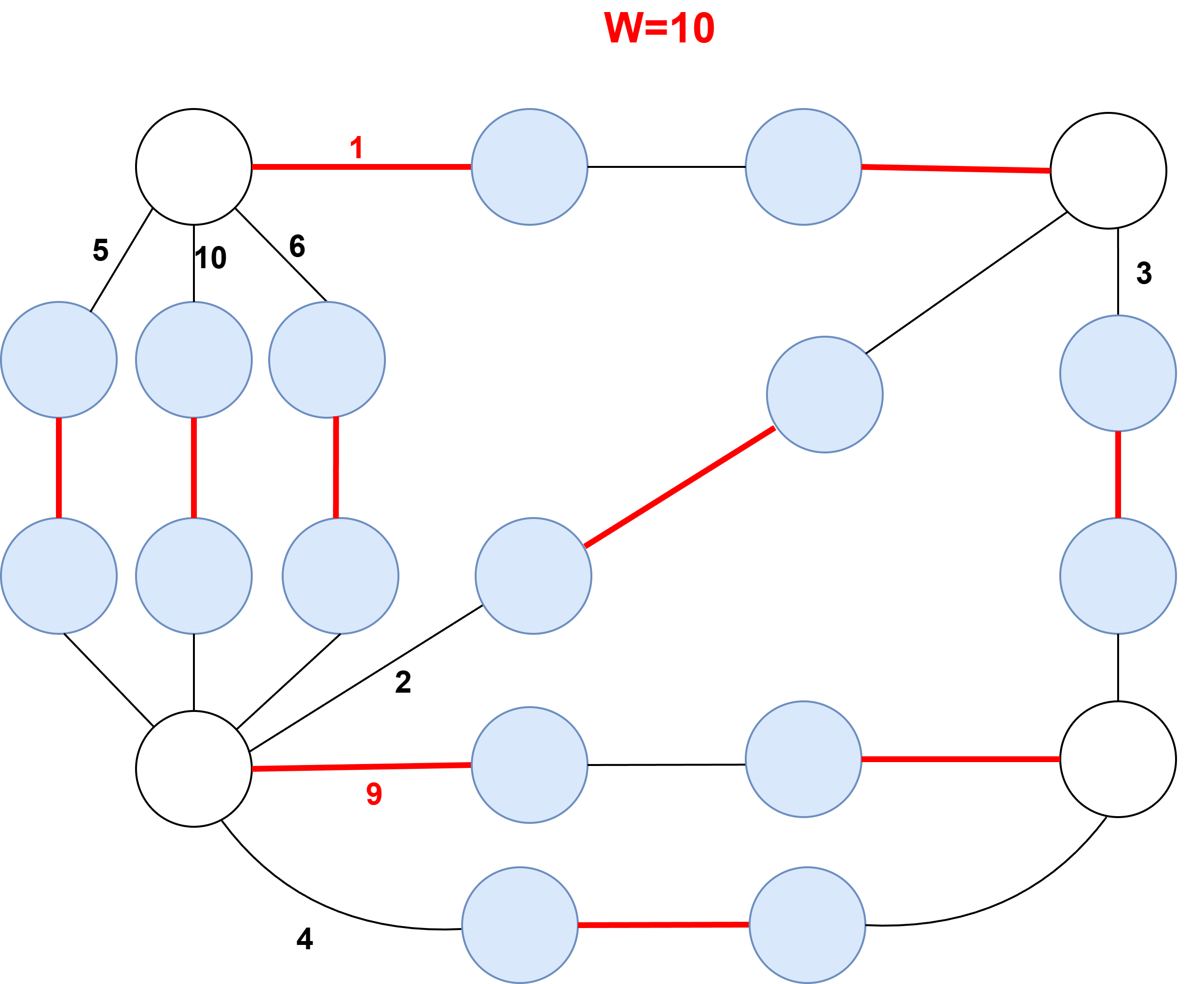}
         \label{fig:1b}
     \end{subfigure}
     \caption{The reduction of Lemma~\ref{lem:multigraph_to_simple}, from \EWPMshort\ in multigraphs to \EWPMshort\ in simple graphs. Unlabeled edges have weight $0$. The perfect matchings of both graphs denoted by red edges have total weight $10$.}
     \label{fig:EWPM}
\end{figure}

\begin{lemma}\label{lem:multigraph_to_simple}
    There is a linear-time reduction from \EWPMshort\ in multigraphs to \EWPMshort\ in simple graphs, preserving weights and multiplying the amount of edges by $3$.
\end{lemma}

\begin{proof}
    Let $G$ be a multigraph.    
    We construct a simple graph $G'$ by replacing every edge of weight~$w$ with a path of length $3$, consisting of one edge of weight $w$ and two edges of weight $0$ (see Figure~\ref{fig:EWPM} for an example).\footnote{The middle edge must have weight $0$. It does not matter which of the other two edges has weight $w$.}
    Observe that for every such path of length~$3$, every perfect matching in $G'$ must either contain the middle edge (weight $0$) or both of the other two edges (weight $w+0=w$).
    We prove that $G$ has a perfect matching of weight $W$ iff $G'$ has a perfect matching of weight $W$.

    ($\Rightarrow$) Let $M$ be a perfect matching in $G$. Construct a matching $M'$ in $G'$ by taking the middle edge for paths of length~$3$ corresponding to edges $e\notin M$ and the other two edges for paths of length~$3$ corresponding to edges $e\in M$. $M'$ is perfect and has the same weight as $M$ by construction.

    ($\Leftarrow$) Let $M'$ be a perfect matching in $G'$. We construct a matching $M$ in $G$ as follows. By the aforementioned observation, it suffices to include an edge in $M$ iff the middle edge of its corresponding path of length $3$ is not in $M'$. $M$ is perfect and has the same weight as $M'$ by construction.
\end{proof}

We obtain the following through Theorems~\ref{thrm:EWPM_to_EM},~\ref{thrm:reduction_PM_to_EWPM} and Lemma~\ref{lem:multigraph_to_simple}.

\begin{theorem}\label{thrm:PM_to_EM}
    \PM\ reduces in $\bigO(n^{p+1}p)$ time to \EM\ with $\bigO(n^{p+1}p)$ vertices and edges.
\end{theorem}

Combining this with Corollary~\ref{cor:old_2V_to_PM_reduction} and Lemma~\ref{lem:PM_self_reduction}, we obtain the main result of this section.

\begin{theorem}\label{thrm:2V_to_EM}
    \twoV\ reduces in $\bigO(n^{p+1}p)$ time to solving an \EM\ instance with $\bigO(n^{p+1}p)$ vertices and edges for each cluster of the input's discretized sequence, where $p$ is the maximum number of distinct deadlines corresponding to a cluster.
\end{theorem}

We obtain the following corollary through the well-known result by Mulmuley et al.~\cite{Vazirani_EM} for \EM\ (see Section~\ref{subsubsec:exact_matching}). 

\begin{corollary}\label{cor:2V_in_RP}
    \twoV\ is in $\mathrm{RP}$ when there is a constant amount of distinct deadlines corresponding to each cluster of the input's discretized sequence.
\end{corollary}

We thus make progress towards the open question in~\cite{kVisits} about the parameterization of \twoV\ by the number of numbers. Although we showed that \twoV\ admits the equivalent of a randomized XP algorithm parameterized by the number of numbers, it still remains open whether it admits a deterministic XP or FPT algorithm with the same parameter.

Note that we do not actually use the number of distinct deadlines of the whole input as parameter; due to Lemma~\ref{lem:PM_self_reduction}, it suffices to use the number of distinct deadlines corresponding to the same \emph{cluster}.

\begin{observation}
\label{obs:bipartite}
    The graphs obtained from Theorems~\ref{thrm:PM_to_EM} and~\ref{thrm:2V_to_EM} are bipartite, since the multigraph of Theorem~\ref{thrm:reduction_PM_to_EWPM} is bipartite, while Lemma~\ref{lem:multigraph_to_simple} and Theorem~\ref{thrm:EWPM_to_EM} both preserve this graph property (due to only replacing edges with paths of odd length). Thus, any algorithm for \EM\ in bipartite graphs is applicable to \twoV\ through Theorem~\ref{thrm:2V_to_EM}.
\end{observation}
Although we do not use Observation~\ref{obs:bipartite} in this work, since the algorithm of Mulmuley et al.~\cite{Vazirani_EM} works for general graphs, we include it as it may be useful in future research.

\section{Density thresholds of \kV\ }\label{sec:density}

In this section we study the density thresholds of the \kV\ problem, for various values of $k$. See Section~\ref{subsubsec:density} for a summary on the concept of density. We use the following terminology:
\begin{itemize}
    \item We refer to the maximum value of density up to which all instances of a problem admit a schedule as the \emph{lower density threshold} of the respective problem.
    \item We refer to the minimum value of density above which no instance of a problem admits a schedule as the \emph{upper density threshold} of the respective problem.
\end{itemize}
Recall that for the infinite version of the problem (i.e., \PWS) the lower density threshold is $5/6$~\cite{Kawamura_5/6_stoc} and the upper density threshold is $1$~\cite{Holte_Pinwheel}. Both of these thresholds are tight.

\subsection{A tight lower density threshold for \oneV\ }

\begin{lemma}\label{lem:lower_density_threshold_oneVisit}
    Every \oneV\ instance $D=\{d_1,\ldots,d_n\}$ with $\mathrm{Dens}(D)\leq 1$ admits a schedule. 
\end{lemma}

\begin{proof}
    We assume $d_1\leq \ldots\leq d_n$, as is standard. Every \oneV\ instance with no feasible schedule must contain a deadline $d_j,\ j \in [n]$, such that $d_j <j$; otherwise, the schedule $1,\ldots,n$ would be feasible.
    We now seek to minimize the density of such an instance. Since both adding more deadlines to $D$ and decreasing the value of a deadline only increases $\mathrm{Dens}(D)$, we set $d_j$ to be the largest element of~$D$, with value $d_j = j-1$. Additionally, the largest value that each $d_i,\ i \in [j-1]$, can take without violating the non-decreasing order of the deadlines is equal to $j-1$, i.e., the value of $d_j$. Hence, the density of a \oneV\ instance that does not admit a schedule cannot be smaller than
    \[
        \sum_{i=1}^j \frac{1}{d_i} = \sum_{i=1}^j \frac{1}{j-1} = \frac{j}{j-1}>1,\; \forall j \in \mathbb{N}.
    \]
    We infer that every \oneV\ instance $D$ with $\mathrm{Dens}(D) \leq 1$ admits a feasible schedule.
\end{proof}

\begin{lemma}\label{lem:density_threshold_tightness_oneVisit}
    For all $\varepsilon > 0$, there exists a \oneV\ instance $D$ with $1<\mathrm{Dens}(D)< 1+\varepsilon$ that admits no schedule.
\end{lemma}

\begin{proof}
    The instance $D=\{n-1,\ldots,n-1\}$ (with $n=|D|$) trivially admits no \oneV\ schedule, since $\max(D) < n$. Its density is $\mathrm{Dens}(D)=n/(n-1)$, which is arbitrarily close to $1$ for large values of $n$.
\end{proof}

We obtain the following theorem from Lemmas~\ref{lem:lower_density_threshold_oneVisit},~\ref{lem:density_threshold_tightness_oneVisit}.

\begin{theorem}\label{thrm:lower_density_threshold_oneVisit}
    \oneV\ has a tight lower density threshold of $1$.
\end{theorem}

\subsection{Bounds for the lower density threshold of \twoV}\label{subsec:density_twoV}

The following lemma is the key to our density bound for \twoV\ and we regard it as our main technical contribution in this section. We begin by taking it for granted and defer its proof to Section~\ref{subsubsec:claimproof}.

\begin{lemma}[Main Lemma]\label{lem:2V_claim}
    Let $D=\{d_1,\ldots,d_n\}$ be a \twoV\ instance with discretized sequence $A=\langle a_1,\ldots,a_n\rangle$, and define $T=[2n]\setminus A=\{t_1,\ldots,t_n\}$.
    If $\mathrm{Dens}(D)\leq \sqrt{2}-1/2$, then it holds that $|T \cap [d_i + a_i]| \geq i$, for all $i \in [n]$.
\end{lemma}

Note that, as usual, we assume that the elements in $D$, $A$ and $T$ are sorted in non-decreasing order. Additionally, by Remark~\ref{remark:2V_deadline_bound} we have $a_n\leq d_n \leq 2n$, which implies $|T|=n$.

\begin{lemma}\label{lem:lower_density_threshold_twoVisits}
    Every \twoV\ instance $D=\{d_1,\ldots,d_n\}$ with $\mathrm{Dens}(D)\leq \sqrt{2}-1/2$ admits a schedule. The respective schedule can be constructed in time $\bigO(n)$.
\end{lemma}

\begin{proof}
Let $D=\{d_1,\ldots,d_n\}$, $d_1 \leq d_2 \leq \cdots \leq d_n$, be a \twoV\ instance with $\mathrm{Dens}(D)\leq \sqrt{2}-1/2$ and let $A=\langle a_1, \ldots, a_n\rangle$ be the discretized sequence of $D$. Define $T=[2n]\setminus A=\{t_1,\ldots,t_n\}$.
By Lemma~\ref{lem:2V_claim}, we have $|T \cap [d_i + a_i]| \geq i$, for all $i \in [n]$. This is directly equivalent to

\begin{equation}\label{eq:claim}
    t_i \le d_i + a_i, \text{ for all } i \in [n].
\end{equation}

We construct a feasible \twoV\ schedule as follows. For each task $i\in [n]$, we place the primary visit of task $i$ in position $a_i$ and the secondary visit of task~$i$ in position $t_i$. By Definition~\ref{def:disc}, it holds that $d_i \geq a_i,\ \forall i\in [n]$, hence the suggested placement of primary visits is feasible. By Equation~\eqref{eq:claim}, the suggested placement of secondary visits is also feasible, given the aforementioned placement of primary visits.
\end{proof}

\begin{theorem}\label{thrm:lower_density_threshold_twoVisits}
    The lower density threshold of \twoV\ is at least $\sqrt{2}-1/2 \approx 0.9142$ and at most~$1$.
\end{theorem}

\begin{proof}
    The former part follows from Lemma~\ref{lem:lower_density_threshold_twoVisits}. Since every instance that admits a \twoV\ schedule also (trivially) admits a \oneV\ schedule, the lower density threshold of \twoV\ cannot be larger than $1$, due to Theorem~\ref{thrm:lower_density_threshold_oneVisit}.
\end{proof}

\subsubsection{Proof of Lemma~\ref{lem:2V_claim}}\label{subsubsec:claimproof}

We will now prove Lemma~\ref{lem:2V_claim}. We split the proof into multiple auxiliary lemmas, as it is quite technical.
Consider the following property, which is the negation of the property mentioned in Lemma~\ref{lem:2V_claim}.

\begin{equation}\label{eq:claim_violation}
    \exists i\in[n]: |T \cap [d_i + a_i]| \leq i-1.
\end{equation}
Let $D^\star = \{ d_1^\star,\ldots d_n^\star \}$ be an instance satisfying~\eqref{eq:claim_violation} with the minimum possible density.
Proving $\mathrm{Dens}(D^\star) > \sqrt{2}-1/2$ directly implies Lemma~\ref{lem:2V_claim}. We will construct $D^\star$ and compute its density.

Let $A^\star= \langle a_1^\star,\ldots,a_n^\star \rangle$ be the discretized sequence of $D^\star$ and $T^\star = [2n]\setminus A^\star$. Let $j \in [n]$ be the smallest index for which $|T^\star \cap [d_j^\star + a_j^\star]| \leq j-1$. We sequentially restrict the form of $D^\star$ as follows.

\begin{lemma}
    \label{lem:step1}
    $|T^\star \cap [d_j^\star + a_j^\star]| = j-1$.
\end{lemma}

\begin{proof}
    Suppose $|T^\star \cap [d_j^\star + a_j^\star]| < j-1$. This can only happen for $j>1$, hence there exists some $d_i^\star \leq d_j^\star$. Suppose we remove $d_i^\star$ from the input; since every number not in $A$ is in $T$ by definition and $a_i^\star \leq d_i^\star$, this causes some number no larger than $d_i^\star$ to be removed from $A^\star$, i.e., be added to~$T^\star$. This only increases $|T^\star \cap [d_j^\star + a_j^\star]|$ by one, maintaining property~\eqref{eq:claim_violation}. Since we removed a deadline from the input, we obtained an instance satisfying~\eqref{eq:claim_violation} with smaller density than that of~$D^\star$, contradicting our assumption for $D^\star$.
\end{proof}

\begin{lemma}
\label{lem:step2}
    The $j-1$ numbers in $|T^\star \cap [d_j^\star + a_j^\star]|$ are the smallest possible, i.e. $T^\star \cap [d_j^\star + a_j^\star] = [j-1]$. Consequently, $j,\ldots, d_j^\star+a_j^\star \in A^\star$.
\end{lemma}

\begin{proof}
    Assume that there exists some $\mu \in [j-1]$ with $\mu\notin T^\star$. This implies the existence of an $\ell \in [j ,\ d_j^\star + a_j^\star]$ with $\ell \in T^\star$, since $|T^\star \cap [d_j^\star + a_j^\star]|=j-1$ by Lemma~\ref{lem:step1}. We remove the minimum deadline from the instance and include a new deadline~$\ell$. The new instance still satisfies~\eqref{eq:claim_violation}, since the cardinality of $A^\star \cap [d_j^\star + a_j^\star]$ (and thus also $T^\star \cap [d_j^\star + a_j^\star]$) remains unchanged. Since the deadline removed from the input is strictly smaller than $\ell$,\footnote{Otherwise it would be $\min(D^\star) \geq \ell$ and $\mu \in A^\star$ with $\mu < \ell$, implying $\ell \in A^\star$ by Def.~\ref{def:disc}, which is a contradiction.} the density decreases, leading to a contradiction.
\end{proof}

\begin{lemma}\label{lem:step3}
    $a_j^\star=2j-1$ and $d_n^\star\geq d_j^\star + 2j -1$.
\end{lemma}

\begin{proof}
    By Lemma~\ref{lem:step2} we have $1,\ldots,j-1 \notin A^\star$ and $j,\ldots,d_j^\star+a_j^\star \in A^\star$, implying that the $j$-th element of $A^\star$ is $a_j^\star=2j-1$. Since $d_j^\star+a_j^\star\in A^\star$, we obtain $d_n^\star\geq a_n^\star \geq d_j^\star+a_j^\star =d_j^\star + 2j -1$.
\end{proof}

\begin{lemma}\label{lem:step4}
    $A^\star=\langle j,\ldots,d_n^\star \rangle$.
\end{lemma}

\begin{proof}
    Since $j\in A^\star$ by Lemma~\ref{lem:step2} and $d_n^\star \in A^\star$ by Definition~\ref{def:disc}, it suffices to show that $A^\star$ consists of a single cluster (see Def.~\ref{def:cluster}). We know that all elements in $[j,\ d_j^\star+a_j^\star]$ are in the same cluster~$C$. By Observation~\ref{obs:cluster}, the discretized sequence of the deadlines of the tasks corresponding to a cluster is the cluster itself. Hence, if there were clusters other than $C$ in~$A^\star$, we could remove from the input the deadlines corresponding to those clusters, thus decreasing the density of the input without falsifying $T^\star \cap [d_j^\star + a_j^\star] = [j-1]$ and $j,\ldots,d_j^\star+a_j^\star \in A^\star$ (due to Obs.~\ref{obs:cluster}). Since these directly imply property~\eqref{eq:claim_violation}, we have reached a contradiction. Therefore, we have $A^\star=C=\langle j,\ldots,d_n^\star \rangle$.
\end{proof}

    Lemma~\ref{lem:step4} implies $n=d_n^\star-j+1$. For $D^\star$ to be density-minimizing, it must hold that:
    \begin{itemize}
        \item The first $j$ deadlines are equal to $d_j^\star$.
        \item The other $d_n^\star-2j+1$ deadlines are equal to $d_n^\star$.
    \end{itemize}
    If either of the above does not hold, replacing the deadlines that do not follow these rules with the ones suggested by them can only decrease the density, leading to a contradiction. 
    
    Observe that the density of the $d_n^\star-2j+1$ largest deadlines is $\frac{d_n^\star-2j+1}{d_n^\star}$, which is increasing with respect to $d_n^\star$. Combining this with Lemma~\ref{lem:step3}, we obtain $d_n^\star=d_j^\star+2j-1$ and thus $n=j+d_j^\star$. Hence, $D^\star$ is the two-level multiset
    \begin{equation*}
    D^\star=\Bigl\{\ \underbrace{d_j^\star,\ldots,d_j^\star}_{j\ \text{copies}}\ ,
    \ \underbrace{d_j^\star+2j-1,\ldots,d_j^\star+2j-1}_{d_j^\star\ \text{copies}}\ \Bigr\},
    \qquad d_j^\star\ge a_j^\star = 2j-1.
\end{equation*}
The density of $D^\star$ is
\[
\mathrm{Dens}(D^\star) 
 = \frac{j}{d_j^\star} + \frac{d_j^\star}{d_j^\star+2j-1}.
\]

This function is minimized when $j = \left[\left(\sqrt{2}-1\right)/2\right] \cdot d_j^\star +1/2$ and $d_j^\star \to \infty$, in which case its infimum equals $\sqrt{2}-1/2$ (see Appendix~\ref{appendix:min_density} for a proof and Figure~\ref{fig:density_function} for a plot of the function). Thus, we obtain $\mathrm{Dens}(D^\star) > \sqrt{2}-1/2$, which proves Lemma~\ref{lem:2V_claim}.

\begin{figure}[ht]
     \centering
     \begin{subfigure}{0.7\textwidth}
         \centering
         \includegraphics[width=\linewidth]{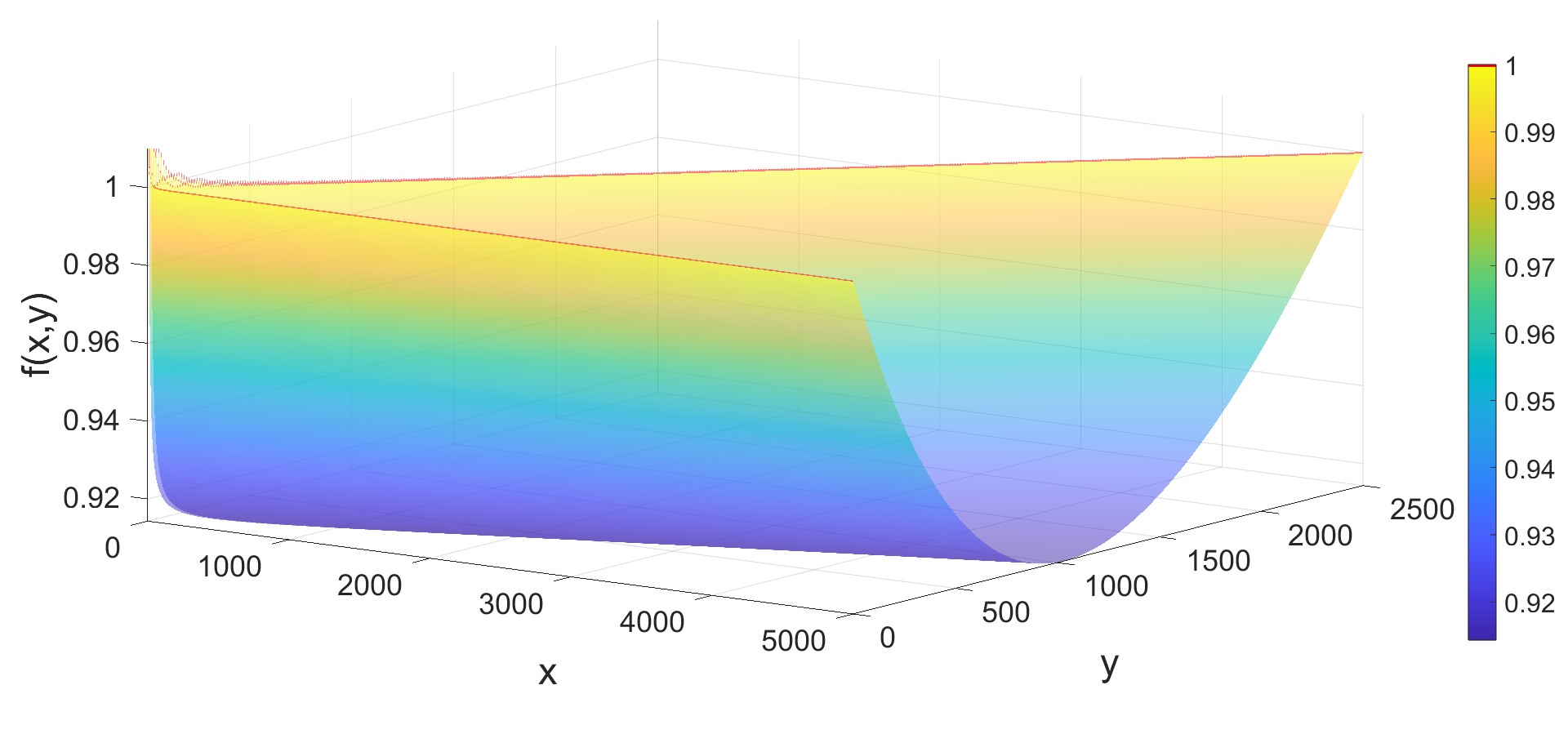}
         \label{fig:2b}
     \end{subfigure}
     \caption{The density function from the proof of Lemma~\ref{lem:2V_claim}: $f(x,y)=\frac{y}{x}+\frac{x}{x+2y-1}$,
subject to $y\ge 1$ and $x\ge 2y-1$. Its infimum is achieved for $y = \frac{\sqrt{2}-1}{2}\cdot x+\frac{1}{2}$ and $x \to \infty$ and is equal to $\sqrt{2}-1/2\approx 0.9142$. See Appendix~\ref{appendix:min_density} for a formal proof.}
     \label{fig:density_function}
\end{figure}

\subsection{Bounds for the lower density threshold of \kV}\label{subsec:upper_bound}

In the following we prove a lemma that connects \kV\ to \PWS\ and use it to obtain an upper bound for the lower density threshold of the former.

Lemma~\ref{lem:finite_to_infinite} is a straightforward modification of the PSPACE-membership proof for \PWS\ by Holte et al.~\cite{Holte_Pinwheel}. We denote the sub-schedule of a schedule~$S$ consisting of the positions from $p$ up to $q\geq p$ as $S[p,\, q]$.

\begin{lemma}[Cyclic schedule]\label{lem:finite_to_infinite}
    Let $S$ be a feasible schedule for a \kV\ instance $D=\{d_1,\ldots,d_n\}$ and let $m=\prod_i d_i$. If for all $i\in [n]$ the $k$-th visit of task $i$ occurs no earlier than position $m+1$, then there exist positions $p,\, p' \in [m]$, $p\leq p'$, such that $S[p,\, p']$ can induce a feasible schedule for the respective \PWS\ instance by infinitely repeating itself.
\end{lemma}

\begin{proof}
    For every position $p\in [m]$, define a \emph{state vector} $V(p)\in \mathbb{N}^{n}$, with the $i$-th element of $V(p)$ representing the remaining time for the deadline of task $i$ to expire, given the schedule $S[1,\, p]$. More formally, define $V(0)=[d_1,\ldots,d_n]$ and then, for all $p\in [m]$, produce $V(p)$ from $V(p-1)$ by decreasing all elements of $V(p-1)$ by $1$, except the one corresponding to the task $i$ visited in position $p$ of $S$, which is instead reset to~$d_i$.

    Since, by assumption, for all $i\in [n]$ the $k$-th visit of task $i$ occurs no earlier than position $m+1$, it follows by Def.~\ref{def:kV} that $S[1,\, m+1]$ contains at least one visit of task~$i$ in any consecutive~$d_i$ positions, for all $i\in [n]$. Note that this does not hold for the entire duration of $S$, since a task does not have to be visited again after being visited $k$ times. We infer that vectors $V(1),\ldots,V(m+1)$ contain only strictly positive elements. Since, by definition, the $i$-th element of $V(p)$ cannot exceed~$d_i$, there are at most $m=\prod_i d_i$ distinct state vectors with strictly positive elements. By the pigeonhole principle, there exist positions $p,\, q \in [m+1]$, $p<q$, such that $V(p)=V(q)$. It is now straightforward that $S[p,\, q-1]$ induces a feasible schedule for the \PWS\ instance $D$ by infinitely repeating itself: $\forall i\in [n]$, it contains at least one visit of task $i$ in any consecutive $d_i$ positions, and this property is infinitely extended because $V(p)=V(q)$, i.e., the remaining time for the deadline of task $i$ to expire is common for positions $p,\, q$, for all $i\in [n]$. Setting $p'=q-1$ completes the proof.
\end{proof}

\begin{theorem}\label{thrm:lower_density_threshold_largeK}
    The lower density threshold of \kV\ is at most $5/6 + 1/\lfloor (k-1)/6 \rfloor$, for all $k\in \mathbb{N}$.
\end{theorem}
    
\begin{proof}
    Consider the instance $D(x)=\{d_1=2,\,d_2=3,\,d_3=x\}$, $x\in \mathbb{N}$, and observe that it is a no-instance of \PWS\ for all $x\in \mathbb{N}$: each of the first two tasks has to be visited once every two time units in an infinite schedule, hence any schedule fails at the $x$-th time unit (cf.~\cite{Chan_conjecture}).

    Set $k=\prod_i d_i +1 = 6x+1$ and assume $D(x)$ admits a \kV\ schedule. Clearly, the $k$-th visit of any task cannot occur earlier than position $k$, thus the conditions of Lemma~\ref{lem:finite_to_infinite} are satisfied for this $k$. We obtain that $D(x)$ is a yes-instance of \PWS, which is a contradiction. We infer that for all $x\in \mathbb{N}$, $D(x)$ is a no-instance of \kV\ with $k= 6x+1$. Since $\mathrm{Dens}(D(x))=5/6+1/x$, this proves the desired statement when $k \equiv 1 \pmod 6$.
    
    For $k \not\equiv 1 \pmod 6$, it suffices to observe that if an instance admits no \kV\ schedule, then it does not admit a $k'$-\textsc{Visits} schedule for any $k'>k$ (cf.~\cite{kVisits}).
\end{proof}

We obtain the following corollary by combining Theorem~\ref{thrm:lower_density_threshold_largeK} with the fact that the lower density threshold of \kV\ cannot be smaller than $5/6$ (see Section~\ref{subsubsec:density} for a short explanation).

\begin{corollary}\label{cor:density_largeK}
    The lower density threshold of \kV\ approaches $5/6$ for $k\rightarrow \infty$, which is known to be tight for \PWS~\cite{Kawamura_5/6_stoc}.
\end{corollary}


\subsection{Nonexistence of upper density threshold for \kV\ }

\begin{theorem}\label{thrm:upper_density_threshold}
    For all $k\in \mathbb{N}$ and all $\delta > 0$, there exists an instance $D$ with $\mathrm{Dens}(D)>\delta$ that admits a \kV\ schedule. Equivalently, there is no upper density threshold for \kV, for any $k\in \mathbb{N}$.
\end{theorem}

\begin{proof}
    Fix $k\in \mathbb{N}$ and consider the instance $D(n)=\{1,\ 1+k,\ 1+2k,\ldots,1+(n-1)k\}$. It is clear that $D(n)$ admits a \kV\ schedule for all $n\in \mathbb{N}$: it suffices to visit the first task $k$ times consecutively, then the second task $k$ times consecutively, and so on. Note that $$\mathrm{Dens}(D(n))=\sum_{i=0}^{n-1} \frac{1}{1+ik}$$ is a harmonic series for $n\rightarrow \infty$, which is a well-known divergent series. Hence, for large values of~$n$, $\mathrm{Dens}(D(n))$ is arbitrarily large, which immediately implies the theorem.
\end{proof}

\section{The \oneortwoV\ problem}\label{sec:one_or_two_visits}

In this section we define and study the \oneortwoV\ problem, a generalization of \twoV.

\begin{definition}[\oneortwoV]\label{def:one_or_two_V}
    Given positive integers $m,n$ and a (multi)set of positive integers (deadlines) $D=\{d_1,\ldots,d_{m+n}\}$, the \oneortwoV\ problem asks whether there exists a schedule of length $m+2n$ visiting exactly one task $i\in [m+n]$ per time unit, such that:
    \begin{itemize}
        \item For each $i\in [m]$, task $i$ is visited exactly once within the first $d_i$ positions of the schedule.
        \item For each $i\in [m+1,m+n]$, task $i$ is visited twice: once within the first $d_i$ positions of the schedule, and once at most $d_i$ positions after its first visit.
    \end{itemize}
\end{definition}

As usual, we assume $d_1\leq \ldots \leq d_m$ and $d_{m+1}\leq \ldots \leq d_{m+n}$.
For tasks $i\in [m+1,m+n]$, we can (equivalently) say that a \emph{primary} visit of $i$ occurs within the first $d_i$ positions of the schedule and a \emph{secondary} visit of $i$ occurs at most $d_i$ positions after its primary visit or at any position before its primary visit, as is standard for \twoV. This alternative definition will allow us to prove a property that disconnects the positions in which primary visits are placed from those in which secondary visits are placed, similar to~\cite{kVisits}. For the tasks that only require one visit, we use the term \emph{single} visits.

It is immediate that \oneortwoV\ is strongly NP-complete, as \twoV\ constitutes its special case with $m=0$. Hence, the focus of this section is to generalize the existing positive results for \twoV\ to \oneortwoV. To this end, we reduce \oneortwoV\ to \PM.

\subsection{Reduction to \PM}

Let $D=\{d_1,\ldots,d_{m+n}\}$ be a \oneortwoV\ input with $m$ tasks that have to be visited once and $n$ tasks that have to be visited twice. For the latter, we will use \emph{primary} and \emph{secondary} visits, as explained above. We assume that $d_i \leq m+2n$ for all $i\in [m+n]$. Note that if there was a deadline larger than $m+2n$ in the input, then all visits of the respective task could be placed at the end of the schedule without affecting its feasibility and, thus, the respective deadline could be removed from the input to obtain an equivalent \oneortwoV\ instance.

Let $A=\langle a_1,\ldots,a_n\rangle$ be the discretized sequence of $\langle d_{m+1},\ldots,d_{m+n}\rangle$. We will prove that it suffices to place the primary visits of tasks $m+1,\ldots,m+n$ in a permutation of the positions in~$A$ (analogous to Lemma~\ref{lem:old_primary_discon_prop}).

\begin{lemma}\label{lem:1or2_primary_discon_prop}
    A \oneortwoV\ instance $D=\{d_1,\ldots,d_{m+n}\}$ admits a schedule if and only if it admits a schedule in which all $n$ primary visits are placed in a permutation of the positions of the discretized sequence $A=\langle a_1,\ldots,a_n\rangle$ of $\langle d_{m+1},\ldots,d_{m+n}\rangle$.
\end{lemma}

\begin{proof}
    The converse direction is trivial; it suffices to prove the forward direction.

    \begin{claim}\label{claim:swap}
        Let $a\in A$. If a feasible schedule places a non-primary (i.e., single or secondary) visit in~$a$, then there exists some primary visit placed (strictly) earlier than $a$, with deadline greater than or equal to $a$.
    \end{claim}

    \begin{claimproof}
        We first prove through induction that at most $n-i$ primary visits can occur (strictly) after position $a_i\in A$.
        \begin{itemize}
            \item \emph{Induction basis $(i=n)$:} By definition, $a_n=d_{m+n}$, hence no primary visit can occur after position $a_n$.
            \item \emph{Inductive hypothesis $(i<n)$:} Suppose that at most $n-i-1$ primary visits occur after position $a_{i+1}$.
            \item \emph{Inductive step $(i<n)$:} By definition, either $a_i=d_{m+i}$ or $a_i=a_{i+1}-1$. In the first case, the primary visits of tasks $m+1,\ldots,m+i$ clearly must be placed no later than position $a_i$, meaning that at most $n-i$ primary visits occur after $a_i$. In the second case, we have at most $n-i-1$ primary visits after position $a_{i+1}$ by the inductive hypothesis. Since $a_i=a_{i+1}-1$, it follows that at most $n-i$ primary visits occur after position $a_i$.
        \end{itemize}

        This concludes the induction.
        
        Now, assume a schedule places a non-primary visit in some position $a_i \in A$. By the above, at most $n-i$ primary visits occur (strictly) after position $a_i$. Since position $a_i$ does not contain a primary visit by assumption, we obtain that at least $i$ primary visits occur (strictly) before position~$a_i$. By the pigeonhole principle, at least one of these $i$ primary visits corresponds to a task with deadline at least $d_{m+i}$. Since $d_{m+i}\geq a_i$ by definition, the claim follows.
    \end{claimproof}

    We will now use Claim~\ref{claim:swap} to apply swapping arguments to an arbitrary feasible schedule in order to prove the lemma.

    Suppose a feasible \oneortwoV\ schedule places a single or secondary visit in a position $a\in A$. By Claim~\ref{claim:swap}, there exists a primary visit placed in some position $p<a$ with deadline at least~$a$, meaning that placing this primary visit in position $a$ maintains feasibility.  It also holds that the single or secondary visit currently occupying position $a$ can be moved to any earlier position without affecting its feasibility\footnote{Note that this only holds for single and secondary visits, not primary visits. Moving a primary visit to an earlier position may cause the respective secondary visit to expire.}. We can thus swap the contents of positions $p$ and $a$, preserving feasibility. Note that it may be the case that the primary visit in $p$ and the secondary visit in $a$ are visits of the same task; however, the aforementioned argument holds even for that case.

    If $p\in A$, then we can apply Claim~\ref{claim:swap} again to move the same single or secondary visit even earlier in the schedule. Since this argument always moves this single/secondary visit to a strictly earlier position, and it keeps holding as long as this single/secondary visit is placed in some position of $A$, applying it iteratively must (at some point) place the single or secondary visit in a position that is not in $A$, since the schedule is finite.
    
    To sum up, we have proven that a single or secondary visit that is placed in a position in $A$ can be moved to an earlier position not in $A$, while only swapping positions with primary visits. Observe that this process does not affect the positions of single/secondary visits other than the one it is applied to; hence, sequentially applying it to each single/secondary visit that is placed in a position of $A$ will result in a schedule in which no single or secondary visit is placed in a position of~$A$. Equivalently, the $n$ primary visits are all placed in some permutation of the $n$ positions of~$A$. This concludes the proof of the lemma.
\end{proof}

With Lemma~\ref{lem:1or2_primary_discon_prop}, we have disconnected the positions in which primary visits should be placed from the positions in which single and secondary visits should be placed. We now disconnect the latter two as well.

\begin{lemma}[Disconnection]\label{lem:1or2_secondary_discon_prop}
    A \oneortwoV\ instance $D\!=\!\{d_1,\ldots,d_{m+n}\}$, with $A=\langle a_1,\ldots,a_n\rangle$ being the discretized sequence of $\langle d_{m+1},\ldots,d_{m+n}\rangle$, admits a schedule if and only if it admits a schedule such that
    \begin{enumerate}
        \item The $n$ primary visits are placed in some permutation of the positions in~$A$.
        \item The $m$ single visits are placed in order of non-decreasing deadline in the latest feasible positions of $[m+2n]\setminus A$.
        \item The $n$ secondary visits are placed in some permutation of the $n$ remaining positions.
    \end{enumerate}
\end{lemma}

\begin{proof}
    By Lemma~\ref{lem:1or2_primary_discon_prop}, we know that the existence of a \oneortwoV\ schedule is equivalent to the existence of a \oneortwoV\ schedule with the first property. In such a schedule, the $m+n$ single/secondary visits are placed in some permutation of the positions in $[m+2n]\setminus A$. The other two properties then follow through swapping arguments. To prove that single visits can be placed in order of non-decreasing deadline, a simple swapping argument to sort them suffices. Additionally, since secondary visits can always be moved to earlier positions without affecting their feasibility (by definition), we can apply swapping arguments to move all single visits as late as possible (in sorted order).
\end{proof}

We are now ready to reduce \oneortwoV\ to \PM.

\begin{theorem}\label{thrm:1or2_PM_reduction}
    \oneortwoV\ reduces to \PM\ in time $\bigO(n+m)$.
\end{theorem}

\begin{proof}
    Let $D=\{d_1,\ldots,d_{m+n}\}$ be a \oneortwoV\ instance, and define $D'=\{d_{m+1},\ldots,d_{m+n}\}$.
    First, compute the discretized sequence $A=\langle a_1,\ldots,a_n\rangle$ of $D'$. Then, for $i=m$ down to $i=1$, place the single visit of task $i$ in the latest available position of $[m+2n]\setminus A$. Lastly, compute the (simple) set $T$ consisting of the $n$ remaining available positions in $[m+2n]\setminus A$. Note that, by Lemma~\ref{lem:1or2_secondary_discon_prop}, the aforementioned placements do not affect the existence of a schedule. Additionally, by Lemma~\ref{lem:1or2_secondary_discon_prop}, a schedule now exists iff each $d\in D'$ can be matched with a position $a\in A$ for its primary visit and a position $t\in T$ for its secondary visit, such that $d\geq a$ and $d+a\geq t$. This is equivalent to the \PM\ instance $(D',T)$, by definition.
    All operations described can be done in time $\bigO(n+m)$.
\end{proof}

\subsection{Algorithms for \oneortwoV}

We are now ready to generalize all known algorithms for \twoV\ to \oneortwoV. Let $A=\langle a_1,\ldots,a_n\rangle$ be the discretized sequence of $\langle d_{m+1},\ldots,d_{m+n}\rangle$.

\begin{theorem}\label{thrm:one_or_twoV_algos}
    \oneortwoV\ admits the following algorithms:
    \begin{enumerate}    
        \item An algorithm running in $\bigO(m+n!)$ time.\footnote{Note that the naive brute force algorithm for \oneortwoV\ runs in time $\bigO((m+2n)!)$.}
        \item An algorithm running in $\bigO(m+n)$ time, if $d_{m+1},\ldots,d_{m+n}$ are distinct numbers.
        \item An algorithm running in $\bigO(m+n)$ time, if there are at most two distinct deadlines corresponding to each cluster of $A$.
        \item An FPT algorithm parameterized by the \emph{maximum cluster size} $c$ of $A$, running in $\bigO(m+n\cdot c!)$ time.
        \item A randomized polynomial-time $\mathrm{(RP)}$ algorithm, if there is a constant number of distinct deadlines corresponding to each cluster of $A$.
    \end{enumerate}
\end{theorem}

\begin{proof}
    (1): By Theorem~\ref{thrm:1or2_PM_reduction}, we can reduce \oneortwoV\ to \PM\ in $\bigO(n+m)$ time. Note that the resulting \PM\ instance consists of three sets of size $n$ each: $D'=\{d_{m+1},\ldots,d_{m+n}\}$, its discretized sequence $A$ and a target set $T$ as described in the proof of Theorem~\ref{thrm:1or2_PM_reduction}. One can run a brute-force algorithm for this \PM\ instance in $\bigO(n!)$ time, computing all feasible matchings between $D'$ and $A$ and then, for each feasible matching, sorting the sums of the respective pairs in non-decreasing order and matching them with the targets in  $T$ in non-decreasing order. It is clear that if no solution exists by matching the sums with the targets in this order, then no solution exists in general. Hence, we obtain an algorithm for \oneortwoV\ running in $\bigO(m+n!)$ time.

    (2): The algorithm follows from Theorem~\ref{thrm:1or2_PM_reduction} and Lemma~\ref{lem:old_distinct}.

    (3): The algorithm follows by sequentially applying Theorem~\ref{thrm:1or2_PM_reduction},  Lemma~\ref{lem:PM_self_reduction} and Lemma~\ref{lem:old_two_numbers}.

    (4): The algorithm follows from Theorem~\ref{thrm:1or2_PM_reduction} and Lemma~\ref{lem:PM_self_reduction}, by applying the brute-force \PM\ algorithm described above for (1) to each of the $\bigO(n)$ instances obtained by Lemma~\ref{lem:PM_self_reduction}. Each of these \PM\ instances has size $\bigO(c)$, and is thus solved in $\bigO(c!)$ time.

    (5): The algorithm follows by sequentially applying Theorem~\ref{thrm:1or2_PM_reduction}, Lemma~\ref{lem:PM_self_reduction}, Theorem~\ref{thrm:PM_to_EM} and the well-known algorithm of Mulmuley et al.~\cite{Vazirani_EM} for \EM\ (see Section~\ref{subsubsec:exact_matching}). 
\end{proof}

\begin{remark}
    The algorithms of Theorem~\ref{thrm:one_or_twoV_algos} coincide with the state-of-the-art algorithms for \twoV\ when $m=0$.
\end{remark}

\section{\threeV\ violates the key structural properties of \twoV\ }\label{sec:3v_counterexample}


In this section we provide a counterexample that shows that Lemma~\ref{lem:old_primary_discon_prop} and the algorithm implied by Lemma~\ref{lem:old_distinct} do not hold for \threeV. This is significant for future work, since the vast majority of results for \twoV\ (both positive and negative) heavily rely on these lemmas; it seems that new methods are yet to be discovered in order to efficiently study \kV\ for $k>2$.

\begin{lemma}\label{lem:3v_discseq_violation}
    There is a yes-instance of \threeV\ for which every feasible schedule contains at least two visits of the same task in the positions of the discretized sequence, violating Lemma~\ref{lem:old_primary_discon_prop}.
\end{lemma}

\begin{proof}
Let $D = \{2,5,6,7,8,9,10,11\}$ be a \threeV\ instance, with discretized sequence (trivially) equal to $D$. 
To keep our analysis concise, we refer to the task with deadline $d \in D$ simply as task~$d$ in the context of this proof.
We begin by observing that the first~$12$ positions of any feasible schedule of $D$ must contain the following $12$ visits.
\begin{itemize}
    \item Three visits to task $2$.
    \item Two visits to each of the tasks $5$ and $6$.
    \item One visit to each of the tasks $7,\ 8,\ 9,\ 10$ and $11$.
\end{itemize}
We infer that no visit can occur in the first $12$ positions, other than the aforementioned $12$ visits (although the order of these $12$ visits is yet to be determined).

Observe that the only visit that can occupy position $12$ is the second visit to task $6$. Therefore, we are forced to schedule the second visit to task $6$ at that position, which consequently forces its first visit to be scheduled at position $6$. Given the above, the only visit that can now occupy position $11$ is the first visit of task $11$, hence we must schedule it there. Additionally, since the first $12$ positions of the schedule can only contain one visit to task $7$, its second visit must be scheduled no earlier than position $13$, meaning that its first visit cannot be scheduled earlier than position $6$. However, position $6$ is already occupied, forcing us to schedule task $7$ in position $7$. Therefore, at this point, the schedule (up to position $12$) has the following form.

\[
\setlength{\arraycolsep}{3pt} 
\begin{array}{r *{12}{c}} 

   \textnormal{Task visits in the first $12$ positions (no order): }& \textcolor{darkorange}{2} & \textcolor{darkblue}{2} & \textcolor{darkgreen}{2} & \textcolor{darkorange}{5} & \textcolor{darkblue}{5} & \textcolor{darkorange}{6} & \textcolor{darkblue}{6} & \textcolor{darkorange}{7} & \textcolor{darkorange}{8} & \textcolor{darkorange}{9} & \textcolor{darkorange}{10} & \textcolor{darkorange}{11}\\[8pt]
    
    \textnormal{Fixed task visits at this point: }& \rule{0.3cm}{1.5pt} & \rule{0.3cm}{1.5pt} 
        & \rule{0.3cm}{1.5pt} & \rule{0.3cm}{1.5pt} & \rule{0.3cm}{1.5pt} & \textcolor{darkorange}{6} & \textcolor{darkorange}{7} & \rule{0.3cm}{1.5pt} & \rule{0.3cm}{1.5pt} & \rule{0.3cm}{1.5pt} & \textcolor{darkorange}{11} 
        & \textcolor{darkblue}{6} 
\end{array}
\]

We now examine what happens at position $10$. Observe that, given the above, this position can only be occupied either by the first visit of task $10$ or by the second visit of task $5$. Suppose that we schedule the first visit of task $10$ at position $10$. Then, the following hold simultaneously.
\begin{itemize}
    \item Since the first visit of task $7$ is scheduled at position $7$ and its second visit does not occur during the first $12$ positions of the schedule, the latter must be scheduled either at position $13$ or $14$.
    \item The third visit of task $5$ must occur after position $12$, thus its second visit cannot be scheduled before position $8$. Since the second visit of task $5$ cannot be scheduled after position $9$ (as positions $10,\ 11$ and $12$ are already occupied), its third visit must be scheduled either at position $13$ or $14$.
    \item The above implies that the second visit of task $5$ must occur either at position $8$ or $9$. As a result, we are forced to schedule either task $8$ or task $9$ before position $6$. Whichever of these tasks is scheduled before position $6$ will need its second visit placed no later than position $9+5=14$, but no earlier than position $13$ due to the observation for the first $12$ positions.
\end{itemize}
From the above, we conclude that there are (at least) three visits that must be scheduled at positions $13$ or $14$, which is impossible. Hence, scheduling task $10$ at position $10$ cannot lead to a feasible \threeV\ schedule. 

The only remaining option is to schedule the second visit of task $5$ at position~$10$, which additionally forces us to schedule its first visit at position $5$. This means that, if a feasible schedule for this instance exists, task $5$ must appear more than once in the positions of the discretized sequence. Thus, it remains to show that such a schedule does in fact exist. Notice that the following is a feasible \threeV\ schedule.

\[
\setlength{\arraycolsep}{3pt} 
\begin{array}{r *{8}{c}} 
   \textnormal{$D$ = $A$ = }& \{2, & 5, & 6, & 7, & 8, & 9, & 10, & 11\}
\end{array}
\]
\[
\begin{array}{r *{24}{c}} 
    \textnormal{$S$ = }& \textcolor{darkorange}{2} & \textcolor{darkblue}{2} 
        & \textcolor{darkgreen}{2} & \textcolor{darkorange}{10} & \textcolor{darkorange}{5} & \textcolor{darkorange}{6} & \textcolor{darkorange}{7} & \textcolor{darkorange}{8} & \textcolor{darkorange}{9} & \textcolor{darkblue}{5} 
        & \textcolor{darkorange}{11} & \textcolor{darkblue}{6} & \textcolor{darkblue}{7} & \textcolor{darkblue}{10} & \textcolor{darkgreen}{5} & \textcolor{darkblue}{8} & \textcolor{darkgreen}{6} & \textcolor{darkblue}{9} & \textcolor{darkgreen}{7} & \textcolor{darkgreen}{8} & \textcolor{darkgreen}{9} & \textcolor{darkblue}{11} & \textcolor{darkgreen}{10} & \textcolor{darkgreen}{11} 
\end{array}
\]

This concludes the proof.
\end{proof}

\begin{remark}
    The violation of Lemma~\ref{lem:old_primary_discon_prop} is significant because all existing positive results for \twoV\ and its generalization, \oneortwoV\ (see Theorem~\ref{thrm:one_or_twoV_algos} for a list of these results), currently rely on Lemma~\ref{lem:old_primary_discon_prop} (and its generalization, Lemma~\ref{lem:1or2_primary_discon_prop}). In fact, even the NP-hardness of \twoV\ strongly relies on Lemma~\ref{lem:old_primary_discon_prop}.
\end{remark}

\begin{lemma}\label{lem:3v_distinct_violation}
    There is a yes-instance of \threeV\ with distinct numbers for which every feasible schedule places the first visits in non-sorted order, violating a property of the algorithm implied by Lemma~\ref{lem:old_distinct}.
\end{lemma}

\begin{proof}
    In the proof of Lemma~\ref{lem:3v_discseq_violation}, after showing that scheduling the task with deadline $10$ at position~$10$ cannot yield a feasible schedule for the \threeV\ instance $D = \{2,5,6,7,8,9,10,11\}$, we are left with a schedule whose first $12$ positions are of the following form.

\[
\setlength{\arraycolsep}{3pt} 
\begin{array}{r *{12}{c}} 

   \textnormal{Task visits in the first $12$ positions (no order): }& \textcolor{darkorange}{2} & \textcolor{darkblue}{2} & \textcolor{darkgreen}{2} & \textcolor{darkorange}{5} & \textcolor{darkblue}{5} & \textcolor{darkorange}{6} & \textcolor{darkblue}{6} & \textcolor{darkorange}{7} & \textcolor{darkorange}{8} & \textcolor{darkorange}{9} & \textcolor{darkorange}{10} & \textcolor{darkorange}{11}\\[8pt]
    
    \textnormal{Fixed task visits at this point: }& \rule{0.3cm}{1.5pt} & \rule{0.3cm}{1.5pt} 
        & \rule{0.3cm}{1.5pt} & \rule{0.3cm}{1.5pt} & \textcolor{darkorange}{5} & \textcolor{darkorange}{6} & \textcolor{darkorange}{7} & \rule{0.3cm}{1.5pt} & \rule{0.3cm}{1.5pt} & \textcolor{darkblue}{5} & \textcolor{darkorange}{11} 
        & \textcolor{darkblue}{6} 
\end{array}
\]

Observe that there are three remaining tasks with deadline greater than $7$ whose first visits have not yet been scheduled and must occur no later than position $12$. However, there exist only two empty positions in the schedule between positions $7$ and $12$. Therefore, at least one task with deadline greater that $7$ must be scheduled before the first visit of task $5$.
\end{proof}

It is straightforward to see that the algorithm implied by Lemma~\ref{lem:old_distinct} (for \twoV\ with distinct deadlines) places primary visits in sorted order in the positions of the discretized sequence and secondary visits in sorted order in the rest of the positions. This implies that first visits are placed in sorted order (by non-decreasing deadline), with the same also holding for second visits. With Lemma~\ref{lem:3v_distinct_violation} we prove that this property is violated for \threeV, reinforcing the following conjecture from~\cite{kVisits}.

\begin{conjecture}
    \threeV\ is strongly NP-complete even when the input is a simple set.
\end{conjecture}

\section{Conclusion}

Our work highlights several interesting properties of \twoV. Although this problem was conceived as a finite version of \PWS, mainly as a stepping stone for studying its complexity, it turns out to be quite interesting in its own right. As highlighted in~\cite{kVisits}, \twoV\ is one of the rare natural problems that is in $\mathrm{P}$ for simple sets, but $\mathrm{NP}$-complete when multisets are allowed. We further enhance this contrast by unexpectedly answering a question of that paper in the negative, proving that \twoV\ is in fact strongly $\mathrm{NP}$-complete even with at most two copies of each number. We find this result both intriguing and surprising, as it contradicts our intuition.

A natural open question arising now is whether \twoV\ can be parameterized by the number of \emph{duplicates}, which could be possible despite its $\mathrm{NP}$-completeness for maximum multiplicity~$2$. Another question remaining open regards the parameterization of \twoV\ by the number of numbers. We conjecture that \PM\ (and thus also \twoV) is FPT parameterized by the number of numbers of $D$, as is standard for numerical matching variants~\cite{Number_of_numbers}. However, as explained in Section~\ref{subsec:related_work}, the standard method of showing this by reduction to \ILP~\cite{Number_of_numbers} fails for \PM.
Our result in Section~\ref{sec:number_of_numbers} indicates that \twoV\ is likely tractable when parameterized by the number of numbers, although the existence of a deterministic XP or FPT algorithm remains open.

Another open question is whether our bound for the lower density threshold of \twoV\ in Section~\ref{sec:density} is tight, as well as whether the $\sqrt{2}-1/2$ number that we obtained 
can provide insight into generalizing our result to \kV\ with $k>2$. We believe that studying the lower density threshold of \kV\ in relation to $k$ is a particularly interesting direction for future work, as it would complement the known $(5/6)$-threshold of \PWS~\cite{Kawamura_5/6_stoc}. Perhaps computer-assisted brute force (like the ones used in~\cite{Gasieniec_towards_5/6,Kawamura_5/6_stoc,Patrolling_SODA_2026} for infinite versions) can also prove useful for the density thresholds of finite versions. 
Since we could not find any unschedulable \twoV\ instance with density smaller than or equal to~$1$, it is possible that the lower density threshold of \twoV\ is exactly~$1$. However, this would imply that \oneV\ and \twoV\ share the same threshold, which seems counterintuitive in conjunction with Corollary~\ref{cor:density_largeK}.

Arguably the most interesting and seemingly difficult direction for future work would be to generalize results of \twoV\ to $k>2$. Currently, none of the algorithms for \twoV\ seem to transfer to \threeV\ and, in fact, even proving the NP-hardness of \threeV\ seems exceptionally challenging, despite strong intuition in favor of it. Our example in Section~\ref{sec:3v_counterexample} highlights the reasons why results from \twoV\ do not currently transfer to \threeV\ and reinforces the conjecture by~\cite{kVisits} that \threeV\ is strongly NP-complete even for simple sets. Transferring (strong) NP-hardness to \kV\ for larger~$k$ would be particularly useful if one could also transfer it to \PWS, which would settle a question open since 1989. This may actually be possible, since the existence of an infinite schedule is known to be equivalent to the existence of an infinite schedule with finite period~\cite{Holte_Pinwheel}. Our Lemma~\ref{lem:finite_to_infinite} establishes another similar connection between finite and infinite versions, further highlighting the relationship between the two problems.

One last open question we would like to highlight is whether \RNTDMshort\ (defined by Yu et al.~\cite{Flow_shop_Yu}) remains strongly NP-complete for simple sets. If this holds, it would simplify the first two steps of our reduction chain in Section~\ref{sec:max_mult_hardness}, as there would be no need to consider the more complicated \NMTSshort\ problem. However, the existing hardness proof for \RNTDMshort~\cite{Flow_shop_Yu} heavily relies on padding \tP\ with duplicate numbers. We believe that settling this question would require a reduction significantly different from the existing one.

\subsubsection*{Acknowledgments}

We are grateful to our professors Aris Pagourtzis and Dimitris Fotakis, as well as our colleague Panos Paskalis, for valuable discussions and insights regarding \PWS, \kV\ and \EM. We are also grateful to professor Euripides Markou, who suggested the \kV\ problem as a research topic.

\medskip

\noindent
This work has been partially supported by project MIS 5154714 of the National Recovery and Resilience Plan Greece 2.0 funded by the European Union under the NextGenerationEU Program.

\medskip

\noindent
No artificial intelligence was used at any point during the production of this paper. All of the presented results were conceived exclusively by the authors and the entirety of the paper was written exclusively by the authors.

\bibliographystyle{splncs04}
\bibliography{bibliography}


\newpage

\appendix

\section{Minimizing the function for the density threshold of \twoV}
\label{appendix:min_density}

We seek to find the minimum (or the infimum) of
\[
f(x,y)=\frac{y}{x}+\frac{x}{x+2y-1},
\qquad\text{subject to } y\ge 1,\ x\ge 2y-1.
\]

First, we must search for critical points of this function in its interior. We compute the partial derivatives

\[
f_y = \frac{1}{x} - \frac{2x}{(x+2y-1)^2}, 
\qquad
f_x = -\frac{y}{x^2} + \frac{2y-1}{(x+2y-1)^2}.
\]
%
%
%
%
%
%
Setting $f_y$ to zero gives us
\begin{align}\label{eq:derivative_y_zero}
    \frac{1}{x} - \frac{2x}{(x+2y-1)^2} = 0 \iff (x+2y-1)^2 = 2x^2.
\end{align}
Substituting this into $f_x$ gives us
\begin{align*}
    -\frac{y}{x^2} + \frac{2y-1}{2x^2} = -\frac{1}{2x^2} < 0.
\end{align*}

Thus, $f$ has no critical point and the minimum (or the infimum) must be achieved on the domain boundary. We examine the following four cases.

\textbf{Case 1:} $y=1$. In this case,
\[
f(x)=\frac{1}{x} + \frac{x}{x+1}, \qquad f'(x)=-\frac{1}{x^2} + \frac{1}{(x+1)^2}<0,\ \textnormal{for}\ x \geq 1.
\]
Thus, $f$ is strictly decreasing and the infimum is achieved for $x\to \infty$ and is equal to $1$.

\textbf{Case 2:} $x = 2y-1$. In this case,
\[
f(y) = \frac{4y-1}{4y-2}, \qquad f'(y) = -\frac{4}{(4y-2)^2} < 0,\ \textnormal{for}\ y \geq 1.
\]
Again, $f$ is strictly decreasing and the infimum is achieved when $y\to \infty$ and is equal to $1$.

\textbf{Case 3:} $x\to \infty$. By Equation~\eqref{eq:derivative_y_zero}, we obtain that the partial derivative $f_y$ is $0$ if and only if
\[
 x = 2y-1 \pm \sqrt{2}(2y-1).
\]
Since $x \geq 2y-1$ and $y\geq 1$ on the domain, only the root with the positive sign is admissible. We obtain $x = 2y-1 + \sqrt{2}(2y-1)$. Solving for $y$, we get
\[
y = \frac{\sqrt{2}-1}{2}\cdot x+\frac{1}{2}.
\]
Since $f_y$ is increasing with respect to $y$, we infer that for every $x\geq 2y-1$ the minimum value of~$f$ with respect to $y$ is achieved for $y = \frac{\sqrt{2}-1}{2}\cdot x+\frac{1}{2}$. Hence, in this case, the infimum is achieved when
\[
x\to \infty, \qquad y = \frac{\sqrt{2}-1}{2}\cdot x+\frac{1}{2}
\]
and is equal to 
\[
\lim_{x \to \infty} \left(  f\left(x,\, \frac{\sqrt{2}-1}{2}\cdot x+\frac{1}{2}\right) \right) =
\lim_{x \to \infty} \left( \frac{\sqrt{2}-1}{2}+\frac{1}{2x} + \frac{1}{\sqrt{2}} \right)=\sqrt{2} - \frac{1}{2}.
\]

\textbf{Case 4:} $y\to \infty$. Since $x\ge 2y-1$, it must also be $x\to \infty$ in this case. Since we have already computed the infimum for $x\to \infty$ in the previous case, this case can be skipped.




The smallest value out of the four cases is obtained from Case 3. Therefore, the infimum of $f(x,y)$ is achieved when $y = \frac{\sqrt{2}-1}{2}\cdot x+\frac{1}{2}$ and $x \to \infty$ and is equal to $\sqrt{2} - 1/2$.

\section{Self-reduction of \PM\ in clusters}
\label{appendix:PM_self_reduction}

In this section we prove Lemma~\ref{lem:PM_self_reduction}, which is implied in~\cite{kVisits}, but not exactly proven in the form that we need in this paper. See Figure~\ref{fig:clusters} for some intuition on why \PM\ self-reduces in clusters.

\setcounter{lemma}{3}

\begin{lemma}[Self-reduction]
    \PM\ reduces in time $\bigO(n)$ to solving a \PM\ instance for the numbers corresponding to each cluster.
\end{lemma}

\begin{proof}
    Let $I=(D=\{d_1,\ldots,d_n\},T=\{t_1,\ldots,t_n\})$ be a \PM\ instance and let $A=\langle a_1,\ldots,a_n\rangle$ be the discretized sequence of $D$. 
    Assume that $A$ consists of two or more clusters, as otherwise the reduction is trivial. 
    
    Let $A'=\langle a_1,\ldots,a_m \rangle$ be the first cluster of $A$. By Def.~\ref{def:cluster}, $a_m < a_{m+1} - 1$. By Def.~\ref{def:disc}, it must hold that $d_m = a_m$ (otherwise it would be $a_m = a_{m+1} - 1$). This implies that $d_m < a_{m+1}$. Recall that $\PM$ does not allow elements in $D$ to be matched with larger elements of $A$. Thus, the elements $d_1,\ldots,d_m$ of $D$ cannot be matched with elements $a_{m+1},\ldots,a_n$ of $A$, implying that they must be matched with some permutation of the elements in $A'$.

    Define $D'=\{ d_1,\ldots,d_m\} $ and $T'=\{t_1,\ldots,t_m\}$ (i.e., the numbers of $D$ and $T$ corresponding to cluster $A'$), as well as $D''=D\setminus D'$, $A''=A\setminus A'$, $T''=T\setminus T'$. At this point we have proven that elements in $D'$ must be matched with elements in $A'$. We will now prove that it suffices to match these with targets in $T'$.
    Suppose that a solution of $I$ contains triplets $(d',a',t'') \in D' \times A' \times T''$ and $(d'',a'',t') \in D'' \times A'' \times T'$. Observe that $t' < t'' \leq d'+a' < d''+a''$; hence, we can substitute these two triplets with $(d',a',t')$ and $(d'',a'',t'')$, obtaining another feasible solution. We conclude that it suffices to match pairs in $D' \times A'$ with targets in $T'$.
    

    By Observation~\ref{obs:cluster}, $A'$ is the discretized sequence of $D'$, i.e., it is one of the sets used for the \PM\ instance $I'=(D',T')$. Similarly, $A''$ is one of the sets used for the \PM\ instance $I''=(D'',T'')$. By all the above, we infer that $I'$ and $I''$ can be solved independently; $I$ has a solution if and only if both $I'$ and $I''$ have a solution.
    
    Iteratively applying the same logic to $I''$ (if $A''$ also consists of two ore more clusters) reduces~$I$ to $\bigO(n)$ $\PM$ instances, one for each cluster of $A$. Then, $I$ has a solution iff each of these $\bigO(n)$ $\PM$ instances have a solution.
    The reduction described runs in linear time.
\end{proof}

\end{document}